\documentclass[10pt,onecolumn,draftclsnofoot]{IEEEtran}

\usepackage{amsmath,amssymb,amsfonts,amsthm,cite,color}
\usepackage[pdftex]{graphicx}
\usepackage{algorithm}
\usepackage{algorithmic}
\usepackage{bm}            
\usepackage[caption=false]{subfig}

\newcommand{\Z}{\mathbb{Z}}
\newcommand{\EX}{\mathbf{E}}
\newcommand{\PR}{\mathbf{P}}

\newcommand{\R}{\mathbb{R}}
\newcommand{\vlambda}{\bm{\lambda}}
\newcommand{\C}{\mathcal{C}}
\newcommand{\bu}{\mathbf{u}}

\newtheorem{theorem}{Theorem}

\newtheorem{lemma}{Lemma}

\newtheorem{definition}{Definition}

\IEEEoverridecommandlockouts

\def\mc{\mathcal}
\def\mbf{\mathbf}

\def\algo#1{
\ifnum#1=1{
  \text{{\sc step\_flex}}
}      
\else \ifnum#1=2{
  \text{{\sc step\_semiflex}}
}
\else{
  \text{{\sc step\_inflex}}
}
\fi
\fi
}

\begin{document}

\title{Efficient and Flexible Crowdsourcing of Specialized Tasks with Precedence Constraints} 
\author{Avhishek~Chatterjee, Michael~Borokhovich, Lav~R.~Varshney, and Sriram~Vishwanath
\thanks{A.~Chatterjee and L.~R. Varshney are with the Coordinated Science Laboratory, University of Illinois at Urbana-Champaign, Urbana, Illinois, USA. (email: \{avhishek,varshney\}@illinois.edu).}
\thanks{M.~Borokhovich is with the AT\&T Labs, New Jersey, USA. (email: michaelbor@utexas.edu).}
\thanks{S.~Vishwanath is with the Wireless Networking and Communication Group, The University of Texas at Austin, Austin, Texas, USA. (email: sriram@austin.utexas.edu).}
\thanks{Part of the material in this paper will be presented at IEEE INFOCOM 2016, San Francisco, USA \cite{ChatterjeeBVV2016}.}
}

\maketitle
\begin{abstract}
Many companies now use crowdsourcing to leverage external (as well as internal) crowds to perform specialized work,
and so methods of improving efficiency are critical.  Tasks in crowdsourcing systems with specialized work have multiple steps and each 
step requires multiple skills. Steps may have different flexibilities in terms of obtaining service from one or multiple agents, due to varying 
levels of dependency among parts of steps. Steps of a task may have precedence constraints among them. Moreover, there are variations in loads 
of different types of tasks requiring different skill-sets and availabilities of different types of agents with different skill-sets. 
Considering these constraints together necessitates the design of novel schemes to allocate steps to agents. In addition, large crowdsourcing 
systems require allocation schemes that are simple, fast, decentralized and offer customers (task requesters) the freedom to choose agents. 
In this work we study the performance limits of such crowdsourcing systems and propose efficient allocation schemes that provably meet the 
performance limits under these additional requirements.  We demonstrate our algorithms on data from a crowdsourcing platform run by a 
non-profit company and show significant improvements over current practice.
\end{abstract}

\section{Introduction}
\label{sec:intro}

The nature of knowledge work has changed to the point nearly all large companies use \emph{crowdsourcing} 
approaches, at least to some extent \cite{Cuenin2015}.  The idea is to draw on the cognitive energy of people, either 
within a company or outside of it \cite{TapscottW2006}. A particularly notable example is the non-profit \emph{impact sourcing} service 
provider, Samasource, which relies on a marginalized population of workers to execute work, operating under the notion \emph{give work, not aid}
\cite{GinoS2012,MarcusP2015}.

There are multifarious crowdsourcing structures \cite{MaloneLD2010,BoudreauL2013} that each require different strategies 
for matching work to agents \cite{DustdarG2011}. Contest-based platforms such as TopCoder and InnoCentive put out open 
calls for participation, and best submissions win prizes \cite{DiPalantinoV2009}.  Microtask platforms such as Amazon Mechanical Turk 
allocate simple tasks on a first-come-first-serve basis to any available crowd agent.  When considering platforms with skilled crowds and
specialized work, such as oDesk (now upWork) \cite{BoudreauL2013}, IBM's Application Assembly Optimization platform \cite{VarshneyACSOLR2014_arXiv},
and to a certain extent Samasource's SamaHub platform \cite{GinoS2012}, efficient allocation algorithms are needed.

In these skill-based crowdsourcing platforms, the specialized tasks have multiple steps, each requiring one or more skills. 
For example, software development tasks may first be planned (architecture), then developed (programming), and finally tested 
(testing and quality assurance), perhaps with several iterations.  Even in skilled microtasking platforms like SamaHub, most 
jobs have more than one step. Task steps often have \emph{precedence constraints} between them, implying that a particular step of a 
task can only be performed after another set of steps has been completed.

To serve a step requiring multiple skills, we need either a single agent that has all of the skills or a group of agents that collectively
do so. Whether multiple agents can be pooled to serve a step or not depends on the \emph{flexibility} of the step: if there are strong 
interdependencies between different parts of a step, the step may require a single agent.  Notions of flexibility and precedence constraints 
are central to this paper.

Allocating tasks to servers is a central problem in computer science \cite{KleinbergT2006}, communication networks \cite{SrikantY2014},
and operations research \cite{Pinedo2012}. The skill-based crowdsourcing setting, however, poses new challenges for task allocation in terms of 
vector-valued service requirements, random and time-varying resource (agents) availability, large system size, a need for simple decentralized 
schemes requiring minimal actions from the platform provider, and the freedom of customers (task requesters) to choose agents without compromising
system performance.  Some of these issues have been addressed in recent work \cite{pang2014service,ChatterjeeVV2015}, but previous work does not address 
precedence constraints or step flexibility.  The notion of flexibility in \cite{ChatterjeeVV2015} is based on agent-categories and is different from here.

Task allocation with precedence constraints has been studied in theoretical computer science, as follows.  Given several tasks, precedence 
constraints among them, and one or more machines (either same or different speed), allocate tasks to minimize the weighted sum of completion 
times or maximum completion time \cite{ChudakS1999}. In crowdsourcing, we have a stream of tasks arriving over time and so we are interested 
in dynamics. 

Dynamic task allocation with precedence constraints 
has recently been studied in  \cite{Pedarsani2015} for Bernoulli task arrivals.
This is different from crowdsourcing scenarios, 
and the optimal scheme is required to search over
the set of possible allocations, which is not suitable for crowdsourcing 
systems due to their inherent high-dimensionality 
(many types of tasks). Additional challenges in a 
crowdsourcing platform are: (i) random and time-varying agent availability; (ii) vector-valued service requirements; 
(iii) fast computation requirements for scaling; and (iv) freedom of choice for customers.

Here we address the above issues for various flexibilities of steps and agents, to characterize limits of crowd 
systems and develop optimal, computationally-efficient, centralized allocation schemes.  Based on insights garnered, we further present
fast decentralized greedy schemes with good performance guarantees. To complement our theoretical results, we also present numerical studies 
on real and synthetic data, drawn from Samasource's SamaHub platform. 

The remainder of the paper is organized as follows. Sec.~\ref{sec:model} describes the system model for crowdsourcing platforms with different 
precedence and flexibility constraints. Sec.~\ref{sec:capacity} presents a generic characterization of the system limits and a generic
centralized optimal allocation scheme. Secs.~\ref{sec:inflexAgentsFlexSteps}--\ref{sec:flexAgentsInflexSteps} address particular 
systems with different flexibility constraints to yield fast decentralized schemes that meet crowdsourcing platform requirements mentioned above.
Sec.~\ref{sec:evaluations} presents numerical studies on real and synthetic data. Detailed proofs of theoretical results are given in 
Appendix~\ref{sec:proofs}.

\section{System Model}
\label{sec:model}
There are a total of $S$ kinds of skills available in the crowdsourcing system, numbered $[S]=\{1, 2, \dots, S\}$.
We define types of agents by skills, and denote the total number of types of agents by $M$. An agent of type $m$ has 
skills $S_m \subset [S]$.

Tasks posted on the platform are of $N$ types. Each type of task $j$ has one or multiple steps associated with it, denoted 
$K_j$. A step $k \in \{1, 2, \dots, K_j\}$ of a job type $j$---a $(j,k)$-step---needs a skill-hour service vector 
$r_{j,k} \in \R_+^{S}$ (non-negative orthant), i.e.\ $r_{j,k,s}$ hours of skill $s$. A part of a step of type $(j,k)$ involving skill $s$ is called a $(j,k,s)$-substep 
if $r_{j,k,s}>0$, the size of this substep.

In the platform, allocations of work to available agents happen at regular time intervals, $t  = 1, 2, \dots$. 
Tasks that arrive after an epoch $t$ are considered for allocation at epoch $t+1$, based on the available agents at that epoch. Tasks
or parts of tasks that remain unallocated due to insufficient available skilled agents are considered again in the next epoch. We assume that 
for any substep $(j,s)$, the time requirement is less than the duration between two 
allocation epochs.

Tasks arrive according to a stochastic process in $\mathbb{Z}_+^N$ (non-negative orthant), $\mathbf{A}(t)=\left(A_1(t), A_2(t), \dots, A_N(t)\right)$, where
$A_i(t)$ is the number of tasks of type $i$ that arrive between epochs $t-1$ and $t$. The stochastic process of available agents at epoch $t$ is 
$\mbf{U}(t)=({U}_1(t), {U}_2(t), \dots, {U}_M(t))$. We assume $\mathbf{A}(t)$ and $\mathbf{U}(t)$ are independent of each other and that each 
of the processes are i.i.d.\ for each $t$, with bounded second moments. Let $\Gamma(\cdot)$ be the distribution 
function of $\mbf{U}(t)$, and let $\bm{\lambda}=\EX[\mathbf{A}(t)]$ and $\bm{\mu}=\EX[\mathbf{U}(t)]$ be the means of the processes.

An agent is \emph{inflexible} if it has pre-determined how much time to spend on each of its skills. Inflexible agents bring a vector 
$\mbf{h}_m = (h_{m,1}, h_{m,2}, \dots, h_{m,S})$ where $h_{m,s} > 0$ if and only if $s \in S_m$ and an inflexible agent spends no more than 
$h_{m,s}$ time for skill $s$.  Contrariwise, \emph{flexible} agents bring a total time $h_m$ which can be arbitrarily split across skills in $S_m$.

A step is \emph{flexible} if it can be served by any collection of agents pooling their service-times. All substeps of \emph{inflexible} steps 
must be allocated to one agent.  At any epoch $t$ only an integral allocation of a step is possible. Hence, in any 
system for a step to be allocated, all of its substeps must be allocated.

A set of flexible substeps $sst_1, sst_2, \dots, sst_n$ of size $x_1, x_2, \dots, x_n$ with skill $s$ can be allocated to agents $1, 2, \dots, m$ if
the \emph{available} skill-hours\footnote{Available skill-hour is determined by the availability of 
the agent, system state, and whether agents are flexible or inflexible.}
 for skill $s$ of these agents, $y_1, y_2, \dots, y_m$, satisfy the following
for some $\{v_{pq} \ge 0:1\le p\le n, 1 \le q \le m\}$,

\begin{equation}
\sum_{p=1}^n v_{ip} \le y_i\ i \in [m],  \sum_{q=1}^m v_{qj} \ge x_j\ j \in [n]\mbox{,}
\label{eq:taskAlloCons1}
\end{equation}
\noindent
where the $\{v_{pq}\}$ capture how agents split their available skill hours across substeps.

For inflexible steps, a set of steps $st_1, st_2, \dots, st_n$ of size (vectors) $\mbf{x}_1, \mbf{x}_2, \dots, \mbf{x}_n$ can be allocated 
to an agent with available skill-hours (vector) $\mbf{y}$ if
\begin{equation}
\sum_{i=1}^n \mbf{x}_i \le \mbf{y} \mbox{.}
\label{eq:taskAlloCons2}
\end{equation}

There may also be precedence constraints on the order in which different steps of a 
task of type $j$ can be served. For any task of type $j$, this constraint 
is given by a directed rooted tree $T_j$ on $K_j$ nodes where a directed edge $(k\to k'), k, k' \in [K_j]$ implies step $k'$ of a task of type $j$ can only be 
served after step $k$ of the same task has been completed.

Scalings of several crowd system parameters are as follows. Task arrival rate $\lambda(N) = \sum_{j=1}^N \lambda_j$ scales faster than number of task types 
$N$, i.e.\ $\lim_{N \to \infty} N/\lambda(N)=0$. Number of skills $S$ scales slower than $N$, i.e.\ $S=o(N)$. In practice, a task
requires a constant number of skills $d$, which implies $\Omega(S^d)$ possible types of 
steps. Number of skills of an 
agent is $d'=O(1)$ implying $M = \sum_l M^l=O(S^{d'})$, implying
$M=O(N^\alpha_1)$ and $M=\Omega(N^{\alpha_2})$ for $0<\alpha_2<\alpha_1$. 
Further, the length of tasks and availability of agents do not scale with 
the size of crowdsourcing systems. 

Beyond these practical system scalings, we make the following mild assumptions: 
$\lambda_j = \omega(1)$ for all $j$ and 
$\lambda^s(N)) = \sum_{j: r_{j,1,s}>0} \lambda_j(N) = \Omega\left(N^c\right)$ $\forall s \in [S]$, for some $c > 0$. These assumptions mean the
arrival rate of every type of job and the total number of jobs requiring a particular skill scale with the system.
We call these scaling patterns \emph{crowd-scaling}.

\section{Notions of Optimality}
\label{sec:capacity}

To formally characterize the maximal supportable arrival rate of tasks, we introduce some more notation and invoke some well-accepted notions used in this regard.

For each $j \in [N]$, let the number of unfinished tasks in the system \emph{just after} allocation epoch $t-1$ be $Q_j(t)$. $A_j(t)$ is the 
number of tasks of type $j$ arriving between epochs $t-1$ and $t$. The number of tasks of type $j$ completely allocated (all steps)
at epoch $t$ is $D_j(t)$. Thus $Q_j(t)$ evolves as:
\begin{equation}
Q_j(t+1) = Q_j(t) + A_j(t) - D_j(t)\mbox{.} 
\label{eq:Qevolution}
\end{equation}
Clearly $D_j(t) \le Q_j(t)+A_j(t)$ at any epoch $t$, since at most $Q_j(t)+A_j(t)$ type $j$ tasks are available. Hence $Q_j(t) \ge 0$ for all $t$. 
Note that due to additional precedence constraints, typically $D_j(t) < Q_j(t)+A_j(t)$. 

\begin{definition}
A scheme of allocation of tasks is called a \emph{policy} if it allocates tasks at a time epoch $t$ based on knowledge of statistics of the processes 
$\mbf{A}$ and $\mbf{U}$ and their realizations up to time $t$, but does not depend on future values.
\end{definition}
\begin{definition}
A crowd system is \emph{stable} under policy $\mathcal{P}$ if the process $\mathbf{Q}(t)=\left(Q_j(t), j \in [N]\right)$ has a finite 
expectation in steady-state under that policy, i.e., $\lim\sup_{t \to \infty} \EX[Q_j(t)] < \infty$, for all $j$ for any initial condition.
\end{definition}
\begin{definition}
An arrival rate $\vlambda$ is \emph{stabilizable} if there exists a policy $\mathcal{P}$ under which $\mathbf{Q}(t)=\left(Q_j(t), j \in [N]\right)$ 
is stable.
\end{definition}
\begin{definition}
The \emph{capacity region} of a crowd system for a given distribution $\Gamma$ of the agent-availability process $\mathbf{U}(t)$ is the closure of the set 
$\mathcal{C}_{\Gamma}=\{\bm{\lambda}: {\vlambda} \ \mbox{is stabilizable}\}$.
\end{definition}
We aim to propose statistics-agnostic, computationally simple and decentralized schemes that offer customers freedom of choice while stabilizing 
any arrival rate in the system's capacity region. Stronger than stability, often we give high probability bounds on number of unallocated tasks.

\section{Capacity and Centralized Allocation Routine}
Here we present a generic characterization of the capacity region of a crowd system for all combinations of agent- and task-flexibility. 
We also give a generic centralized allocation routine that can be easily adapted to a particular system.

For any given set of available agents $\bu=\left(u_i: 1 \le i \le M \right)$, 
define the number of different types of steps ($\{a_{j,k}\}$) that can \emph{potentially} be allocated in a crowd system by $C(\bu) \subset \R_+^{\sum_j K_j}$.
When we say $\{a_{j,k}\}$ is the number of steps of different types that can potentially be allocated, we consider the following scenario
that satisfies the allocation constraints in Sec.~\ref{sec:model}.
\begin{enumerate}
\item[A1] An infinite number of steps of each type $(j,k)$, $k \in [K_j]$  for a $j \in [N]$ are available for allocation, i.e., the limitation only 
comes from the available resource $\bu$.

\item[A2] Precedence constraints among the steps are already satisfied, i.e., all corresponding $(j,k)$-steps of the available $(j,k')$-steps have already 
been allocated previously.  This is equivalent to an absence of precedence constraints.

\item[A3] Integral steps must be allocated, i.e., all substeps of a step need to be allocated for allocation of the step.

\item[A4]  To allocate $a_{j,k}$ steps of different types to a collection of $R$ agents of type $\{m_r: r \in [R]\}$ and available hours 
$\{y_{m_r,s}:r \in [R]\}$ (which is a function of $\{\mbf{h}_{m_r}: r \in [R]\}$ depending on the system), we need to satisfy either 
\eqref{eq:taskAlloCons1} or \eqref{eq:taskAlloCons2} depending on system type.
\end{enumerate}

Let $C^{\text{cvx}}(\bu)$ be the convex hull of the set $C(\bu)$, and define another set $\mbf{C} \subset \R_+^{\sum_j K_j}$ as follows.
\[
\mbf{C} = \left\{\sum_{\bu} \Gamma(\bu) \mbf{a}(\bu) : \mbf{a}(\bu) \in C^{\text{cvx}}(\bu)\right\}
\]
Based on this we define another set $\mc{C} \subset \R_+^{N}$. Let for any $\mbf{a} \in \R_+^N$, 
${\mbf{a}}^E := ((a_1, a_1, \ldots, K_1 \mbox{ times}), (a_2, a_2, \ldots, K_2 \mbox{ times})$, 
$\ldots, (a_N, a_N, \ldots, K_N \mbox{ times}))$.  Then $\mc{C} = \{\mbf{a}: {\mbf{a}}^E \in \mbf{C}\}$. 
This set characterizes the capacity region of the crowd system.
\begin{theorem}
\label{thm:capacity}
Any arrival rate $\vlambda$ is stabilizable if for some $\epsilon > 0$,
$\vlambda + \epsilon \mbf{1} \in \mc{C}$ and no arrival rate $\vlambda$
can be stabilized if $\vlambda$ is outside the closure of the set $\mc{C}$.
\end{theorem}

Note that we ignore the precedence constraint in defining $C(\bu)$. This
does not conflict with the fact the capacity region is a subset
of $\mc{C}$, but it may not be obvious $\mc{C}$ is
in fact the capacity region. \emph{A fortiori}, we show this with a scheme
that respects precedence constraints and stabilizes any rate in the interior of 
$\mc{C}$. 
\subsection{Centralized Allocation}
\label{sec:centralized}
Let us develop a statistics-agnostic scheme that stabilizes any arrival rate $\vlambda$.

Let $Q_{j,k}(t)$ be the number of unallocated $(j,k)$ steps just before allocation epoch $t$. This includes steps not allocated at epoch $t-1$ and steps 
that became available for allocation between $t-1$ and $t$. Thus, 
if for any $(j,k)$, $D_{j,k}(t)$ $(j,k)$-steps were allocated at epoch $t$ and 
$A_{j,k}(t+1)$ new $(j,k)$-steps became available between $t$ and $t+1$, 
\[
Q_{j,k}(t+1) = Q_{j,k}(t)-D_{j,k}(t)+A_{j,k}(t+1)\mbox{.}
\]
Note that, for any $j$ and $K_j \ge k >1$, new $(j,k)$-steps become available only when some $(j,k-1)$ steps have been completed. 
Service times $\{r_{j,k,s}\}$ are strictly less than the duration between two allocation
 epochs. So, any step allocated at epoch $t$ 
is completed before epoch $t+1$. Hence, for any $j$ and $K_j \ge k >1$: $A_{j,k}(t+1) = D_{j,k-1}(t)$.
On the other hand, for any $j$ and $k=1$, we have an external arrival $A_j(t+1)$ between epoch $t$ and $t+1$.

At any time $t$, for a given resource availability, an allocation rule determines resources to be allocated for certain number of
$(j,k)$-steps. We denote this by $S_{j,k}(t)$. Note that $D_{j,k}(t) = \min(Q_{j,k}(t), S_{j,k}(t))$.
Our goal is to design a scheme that finds a good $\{S_{j,k}(t)\}$ for a given $\{Q_{j,k}(t)\}$ and $\mbf{U}(t)=\bu$. 

\begin{center}
{\bf Centralized Allocation} 
\vspace{-0.1 in}
\begin{algorithmic}[1]
Input: $\{Q_{j,s}(t):j \in [N],s \in [S]\}$ and $\mathbf{U}(t)$ at $t$\\
Output: $\{S^*_{j,k}(t)\}$ and allocation of steps to agents
\STATE Define: $l_{j,r}:$  number of leaves in the subtree of $T_j$ rooted
at $r$ 
\STATE Obtain $\{S^*_{j,k}(t)\} = \arg \max_{s_{j,k} \in C(\mbf{U}(t))}$
\begin{align}
&  \sum_{j} \sum_{k=1}^{K_j}  
\sum_{r:k\to r \in T_j}
s_{j,k} l_{j,r} (Q_{j,k}(t) - Q_{j,r}(t)) \label{eq:centAllo}
\end{align}
\STATE For each $(j,k)$ allocate $S^*_{j,k}(t)$ $(j,k)$-steps 
\end{algorithmic}
\end{center}
This allocation scheme is statistics-agnostic and explicit in terms of system state. 
Also, note that by the design of the scheme the precedence constraint is 
automatically satisfied.
One important thing to note is that the allocation scheme is generic, 
in the sense that this policy can be easily adapted for different
agent- and step-flexibility. 
Note that $C(\mbf{U}(t))$ comes from the allocation constraints of the system. If
in \eqref{eq:centAllo} we replace $C(\mbf{U}(t))$ by the corresponding allocation set, 
the centralized algorithm becomes a generic allocation routine. 

In fact, the generic statistics-agnostic routine for centralized allocation scheme described above is optimal, in the sense that any arrival rate that can 
possibly be stabilized by any policy can also be stabilized by this scheme. 
\begin{theorem}
\label{thm:centralized}
The centralized allocation routine described above stabilizes any $\vlambda$ if $\vlambda + \epsilon \mbf{1} \in \mc{C}$, the capacity region
of the corresponding system for any $\epsilon>0$.
\end{theorem}
Though the scheme has similarity with back-pressure algorithms \cite{TassiulasE1992,Neely2010,SrikantY2014}; unlike the back-pressure scheme it also uses graph parameter ($l_{j,r}$) in computing the weights. Proof is using a Lyapunov function involving $\{l_{j,r}\}$ and queue-lengths.

Instead of directed rooted tree if the precedence constraint is a directed acyclic graph the same results extend. It would be apparent from proof of Theorem \ref{thm:capacity} that the converse (outer-bound on capacity) depend on the precedence graph. On the other hand, for any precedence constraint given by a directed acyclic graph, there exists a precedence constraint given by a directed rooted tree such that the tree constraint does not violate the directed acyclic graph constraint. Then, by applying the above centralized algorithm for this directed rooted tree capacity can be achieved.

\section{Inflexible Agents and Flexible Steps}
\label{sec:inflexAgentsFlexSteps}
Here we characterize the limits of tasks allocation where all steps are flexible and agents are inflexible. Sec.~\ref{sec:capacity} presented 
a generic capacity characterization and algorithm; this section investigates computational aspects of the generic algorithm for this particular system
and also proposes a simple decentralized scheme that works well under a broad class of assumptions.

Consider $C_{I,F}(\bu)$, the set of possible allocations with inflexible agents $(I)$ and flexible tasks $(F)$ for availability of agents, $\bu$. 
Recall the allocation scenario in Sec.~\ref{sec:capacity} to determine a generic $C(\bu)$: A1--A3 are the same for
any system flexibility, but A4 is specific.  For an $(I, F)$ system we have the following.

To allocate $a_{j,k,s}$ tasks of different types to a collection of $R$ agents
of type $\{m_r: r \in [R]\}$ and available hours $\{h_{m_r,s}:r \in [R]\}$ we must satisfy \eqref{eq:taskAlloCons1}:
\begin{equation}
\sum_{j,k} a_{j,k,s} r_{j,k,s} \le \sum_{r=1}^R h_{m_r,s} \mbox{ for all } s \in [S] \mbox{.}
\label{eq:inflexibleCons}
\end{equation}
Note that whenever a step is allocated, all tasks in it must be allocated simultaneously. Hence, we can only allocate 
$a_{j,k,s}$ tasks with $a_{j,k,s}=a_{j,k,s'} \ \forall  s, s' \in [S]$ when satisfying \eqref{eq:inflexibleCons}.

Given $C_{I,F}(\bu)$, the capacity region $\mc{C}_{I,F}$ is obtained in the same way $\mc{C}$ was obtained from $C(\bu)$ in 
Sec.~\ref{sec:capacity}.

The generic centralized allocation routine can be similarly specialized for $(I,F)$ systems: $C(\mbf{U}(t))$ in \eqref{eq:centAllo} of the routine is 
replaced by $C_{I,F}(\mbf{U}(t))$.  The centralized scheme is computable since $C_{I,F}(\mbf{U}(t))$ can be written explicitly in terms of 
$\mbf{U}(t)$, $\mbf{r}_{j,k}$, and $\mbf{h}_m$, but it cannot always be computed in polynomial time.  Since any allocation in $C_{I,F} (\bu)$ must 
satisfy constraint \eqref{eq:inflexibleCons}, optimization problem \eqref{eq:centAllo} can be written as:
\begin{align}
&\max_{s_{j,k} \in \Z_+} \sum_{j} \sum_{k=1}^{K_{j-1}} w_{j,k} s_{j,k} \ \nonumber \\
&\quad\mbox{s.t. } \sum_{j,k} s_{j,k} r_{j,k,s} \le \sum_m u_m h_{m,s} \mbox{ for all } s \in [S] \mbox{.}
\label{eq:centAlgo}
\end{align}
Note that the solution to the problem does not change if we replace $w_{j,k}$ by $\max(v_{j,k},0)$, as optimal schemes never allocate resources to negative $w_{j,k}$. 
Thus, we assume $w_{j,k} \ge 0$. 

Note that \eqref{eq:centAlgo} is a multi-dimensional knapsack problem, where the number of available items of a given weight and value are unbounded 
\cite{KellererPP2004}. This problem is known to be NP-hard and without any fully polynomial-time approximation scheme (FPTAS). 
A polynomial-time approximation scheme (PTAS) is known, but the complexity is exponential in dimension. Recently extended linear
programming (LP) relaxations have been proposed, but have the same issues (see 
\cite{Pritchard2009} and references therein). 

We aim to find a simple and fast distributed scheme, but first propose the following 
LP relaxation-based, polynomial-time (in $N$ and $M$) scheme that gives nearly 
optimal centralized allocation for a large crowd system (under crowd scaling).
\begin{align}
&\{\hat{S}_{j,k}(t)\}= \mbox{\eqref{eq:centAlgo} with relaxing of} 
\ \{s_{j,k} \in \R_+\} \nonumber \\
&\mbox{Allocate } \{{S}^r_{j,k}(t)=\lfloor \hat{S}_{j,k}(t) \rfloor\} 
\label{eq:centAlgoLP} 
\end{align}

We cannot give performance guarantees for this scheme at each allocation epoch for arbitrary $Q_{j,k}$, but for a sufficiently large crowd system, this scheme
stabilizes almost any arrival rate that can be stabilized.
\begin{theorem}
\label{thm:LPrelax}
Under crowd scaling, for any $\alpha < 1$ there is an $N_0$ such that for any system with $N \ge N_0$, the LP-based scheme \eqref{eq:centAlgoLP} stabilizes
any arrival rate in $\alpha \mc{C} = \{\mbf{a}: \frac{\mbf{a}}{\alpha} \in \mc{C}\}$.
\end{theorem}

\subsection{Decentralized Allocation}
\label{sec:prioGreedy}
In this section we develop a simple decentralized scheme with good performance guarantees. As discussed before, often one of the main reasons 
for customers to go to a crowd platform is the ability to choose workers themselves. As such, we propose a simple greedy scheme that allows 
customers the freedom of choice with minimal intervention from platform operators.  This also reduces the platform's operational cost.

In greedy allocation, each step competes against others to find an allocation for all of its tasks. Contention can be resolved arbitrarily, e.g., 
random, pre-ordered, or age-based. 

The \emph{Prioritized Greedy} algorithm below performs greedy allocation among all steps across all types of tasks that are in the same order. It starts with
steps that are in the beginning of the precedence tree and once these steps find an allocation (or cannot be allocated), only then are steps lower in the 
corresponding precedence trees allowed to allocate themselves.

\begin{algorithm}[h]
\caption{{Prioritized Greedy}} \label{alg:prioGreedy}
Define $D=\max_j \mbox{depth of} \ T_j$

\begin{algorithmic}[1]
\STATE $\mc{S}_j=\emptyset \mbox{ for all } j \in [N]$
\FOR{d=1:D}
\STATE $\mc{S}_j = \{k_j: \mbox{ depth of } k_j \mbox{ in } \ T_j = d\}$ 
\STATE Greedy allocation among $\cup_j \{j,k_j: k_j \in \mc{S}_j\}$ steps 
\ENDFOR
\end{algorithmic}
\end{algorithm}

This algorithm can be efficiently implemented on a crowdsourcing platform with minimal intervention from the platform operator.
The operator need only tag unallocated steps in the system based on their depth in the rooted precedence tree and only show available workers
to them after steps at lower depth have exercised their allocation choice. This may be implemented by personalizing the platform's
search results.

The algorithm is fast and has good performance guarantees under certain broadly-used assumptions on arrival and availability processes.
\begin{definition}
A random variable $X$ is \emph{Gaussian-dominated} if $\EX[X^2] \le \EX[X]^2+\EX[X]$ and for all $\theta \in \R$, 
$\EX[e^{\theta \left(X-\EX[X]\right)}] \le e^{\frac{1}{2}(\left(\EX[X^2]-\EX[X]^2\right) \theta^2}$, and \emph{Poisson-dominated} if for all $\theta \in \R$, $\EX[e^{\theta \left(X-\EX[X]\right)}] \le e^{\EX[X] (e^\theta-\theta-1)}$.
\end{definition}
These domination definitions, commonly assumed in 
bandit problems \cite{BubeckC2012}, imply that variation around the mean is dominated in a moment generating function sense  by that of a Gaussian (Poisson) 
random variable. Such a property is satisfied 
by many distributions used to model arrival processes, including in crowdsourcing systems \cite{VukovicS2012}.

\begin{theorem}
\label{thm:prioGreedy}
Consider inflexible agents and flexible steps crowdsourcing systems (size $N$) where for any $s,s'$ $|\sum_m \mu_m h_{m,s} - \sum_m \mu_m h_{m,s'}|$ is sub-poly$(N)$, i.e., $o(N^\delta), \forall \delta>0$, arrival and availability processes are Poisson-dominated (and/or Gaussian-dominated), and system scales as per crowd-scaling.
Then, for any $\alpha \in (0,1)$, $\exists N_\alpha$ s.t. $\forall N\ge N_\alpha$, any arrival rate $\vlambda \in \alpha \mc{C}_{I,F}$ can be stabilized by Prioritized Greedy, and at the steady state the total number of unallocated steps in the system across all types is $O(\log N)$ w.p. $1-O\left(\frac{1}{N^2}\right)$.
\end{theorem}
This implies Prioritized Greedy can stabilize almost any stabilizable arrival rate 
while ensuring the number of unallocated tasks scales more slowly than the system size.

\section{Flexible Agents and Flexible Steps}
\label{sec:flexAgentsInflexSteps}
Now consider systems with flexible agents and flexible steps $(F,F)$, and characterize capacity regions. For a given availability of agents $\bu$, the 
set of possible step allocations are $C_{F,F}(\bu)$. As for $C(\bu)$ in Sec.~\ref{sec:capacity} this satisfies A1--A3 in the allocation scenario;
A4 for $(F,F)$ systems is as follows.

A certain number of steps $\{a_{j,k}\}$ of each type can be allocated to a set of agents $\{1, 2, \dots, R\}$ of types $\{m_r: r \in [R]\}$ if there exists 
a set of $R$ vectors in $\R^S$, $\tilde{h}_{m_r} = (\tilde{h}_{m_r,s} \ge 0: s \in [S])$, such that:
\begin{align}
&\tilde{h}_{m_r,s}=0 \mbox{ if } s \not\in S_{m_r} \mbox{ for all } s, m_r \mbox{,}\nonumber \\
&\sum_s \tilde{h}_{m_r,s} \le \tilde{h}_{m_r} \mbox{ for all } r \in [R]\mbox{, and} \nonumber  \\
&\sum_{j,k} a_{j,k} r_{j,k,s} \le \sum_{r=1}^R \tilde{h}_{m_r,s}  \mbox{ for all } s \in [S]\mbox{.} 
\label{eq:flexibleCons}
\end{align}
Based on the set of possible allocations $C_{F,F}(\bu)$, the capacity region $\mc{C}_{F,F}$ can be characterized just as in Sec.~\ref{sec:capacity}.

Similar to Sec.~\ref{sec:inflexAgentsFlexSteps}, if we replace $C(\mbf{U}(t))$ by $C_{F,F}(\mbf{U}(t))$ in the centralized allocation routine we obtain an 
optimal policy for the $(F,F)$ system. It is not hard to see that for the instance where each agent has exactly one skill, the problem is again a 
multi-dimensional knapsack problem and therefore NP-hard. We develop a computationally-efficient scheme. 

If there are $R$ agents of type $\{m_r: r \in R\}$ available, then the centralized allocation problem at time $t$ is to optimize:
\begin{align}
&\max_{(s_{j,k} \in \Z_+), (\tilde{h}_{m_r,s} \in \R_+)} \sum_{j,k} w_{j,k} s_{j,k} \mbox{ s.t. constraints in (\ref{eq:flexibleCons})}. 
\label{eq:RvarOpti}
\end{align}
This is a mixed ILP with $\sum_{j=1}^N K_j$ integer variables and $RS$ real variables. The complexity of this problem scales with the number
of available agents in the system, $R$. We would like to avoid such a scaling as $R$ may be much larger than $M$ and $N$ in a crowd system. 
Hence, we pose another optimization problem where the number of variables scales with $M$ and $N$. 

Given $U(t)=\bu$,
\begin{align}
&\max_{(s_{j,k} \in \Z_+), (\alpha_{m,s} \in \R_+)} \sum_{j,k} w_{j,k} s_{j,k} \nonumber \\
&\mbox{s.t. } {\alpha}_{m,s}=0 \mbox{ if } s \not\in S_{m} \mbox{ for all } s, m\mbox{,}\nonumber \\
&\quad\sum_s {\alpha}_{m,s} \le 1 \mbox{ for all } m \in [M]\mbox{, and} \nonumber  \\
&\quad \sum_{j,k} s_{j,k} r_{j,k,s} \le \sum_{m=1:M} u_m h_m \alpha_{m,s}  \mbox{ for all } s \in [S]\mbox{.} 
\label{eq:MvarOpti}
\end{align}
Note that this optimization problem yields an allocation satisfying all constraints for flexible agent allocation. This is because $\alpha_{m,s}$
is the fraction of time of an agent of type $m$ that has been given to skill $s$, which can be positive only when agent of type $m$ has skill $s$. 
The last inequality ensures that the skill-hour constraint per skill is satisfied. Hence, this is a feasible allocation procedure.

This is again a mixed ILP, but with $M+N$ variables. Note that this problem is also NP-hard, corresponding to a multi-dimensional knapsack
problem if $|S_m|=1, \mbox{ for all } m \in [M]$. We design a centralized scheme that allocates steps based on
the following LP relaxation. Given $U(t)=\bu$,
\begin{align}
&(\hat{s}_{j,k}:j,k) = \arg\max_{(s_{j,k} \in \R_+), (\alpha_{m,s} \in \R_+)} \sum_{j,k} w_{j,k} s_{j,k} \nonumber \\
& \mbox{s.t. constraints in (\ref{eq:MvarOpti})} 
\mbox{ and allocate } \{\lfloor \hat{s}_{j,k} \rfloor\} \mbox{ steps.} 
\label{eq:MvarOptiRela}
\end{align}
This scheme has the following performance guarantee.
\begin{theorem}
\label{thm:LPrelaxFlex}
Under crowd scaling, for any $\alpha < 1$ there is an $N_0$ s.~t.\ for any system with $N \ge N_0$, the LP-based scheme \eqref{eq:MvarOptiRela} 
stabilizes any arrival rate in  $\alpha \ \mc{C}_{F,F} = \{\mbf{a}: \frac{\mbf{a}}{\alpha} \in \mc{C}_{F,F}\}$.
\end{theorem}
Proof of this theorem is based on the equivalence of 
\eqref{eq:RvarOpti} and \eqref{eq:MvarOpti}.

\subsection{Decentralized Allocation}
\label{sec:decentralizedFlex}
Now we develop a decentralized allocation scheme that requires minimal centralized operation, and gives customers the option
to choose from a pool of multiple agents.

\begin{algorithm}[h]
\caption{Prioritized Greedy with Flexibility} \label{alg:flexGreedy}
Initialize: $\{\gamma(t-t_0)\in [0,1], t\ge t_0\}$, 
at starting time $t_0$ $\bar{A}(t_0)=\mbf{1}$, $\epsilon>0$

\begin{algorithmic}[1]
\STATE Update at each $t$:
\begin{align}
\bar{\mbf{A}}(t) & = (1-\gamma(t-t_0)) \bar{\mbf{A}}(t-1) + \gamma(t-t_0) \mbf{A}(t) 
\nonumber \\
\bar{\mbf{U}}(t) & = (1-\gamma(t-t_0)) \bar{\mbf{U}}(t-1) + \gamma(t-t_0) \mbf{U}(t)
\nonumber 
\end{align}
\STATE Solve for $\{\psi_{m,s}(t):s \in [S], m \in [M]\}$
\begin{align}
& \ \ \ \max 1 \ \mbox{s.t.} \nonumber \\
& \sum_{j,k} \bar{A}_j(t) r_{j,k,s} \le (1-\epsilon)
\sum_{m:s \in S_m} \bar{U}_m(t) \psi_{m,s}(t) \nonumber \\
& \psi_{m,s}(t)\ge 0, \sum_s \psi_{m,s}(t) \le 1 \nonumber \\
& \psi_{m,s}(t)>0 \ \mbox{only if} \ s \in S_m, \label{eq:flexGreedyAlg1}
\end{align}
if no solution pick $\psi_m(t)$ randomly from a simplex in 
$\R^S$.

\STATE Initialize sets: $\mc{B}_s = \emptyset, \forall s \in [S]$
\STATE For each type $m$: put each available agent in
one of $\{\mc{B}_s\}$ w. p. $\{\psi_{m,s}(t)\}$ (independent rolls of loaded dices)
\STATE Create {\em inflexible} agents:
an agent of type $m$ in $\mc{B}_s$ has $h_m$ available time only for
skill $s$
\STATE Run {\bf Prioritized Greedy} for this (I,F) system
\end{algorithmic}
\end{algorithm}

This algorithm is amenable to crowdsourcing platform implementation. Note 
$\bar{\mbf{A}}(t)$ is available from recent history. Creating the set $\mc{B}_s$ is simple:
for any agent of type $m$ we just randomly tag (as per $\psi$) with a particular skill and it
is shown only tasks with this particular skill. Similarly customers are
only shown that the agent has only the particular skill. The rest of the algorithm is exactly like
Prioritized Greedy where we create classes of steps and priorities among them and then within
each class the allocation is arbitrarily greedy. 

We can guarantee Alg. \ref{alg:flexGreedy} performance when $\gamma$ satisfies:
$\gamma(x)=O\left({x}^{-1}\right)$ and $\gamma(x)=\Omega\left({x^{-\frac{1}{2}+\epsilon}}\right), \epsilon>0$.
\begin{theorem}
\label{thm:flexGreedy}
Consider a flexible agents and flexible steps crowdsourcing system with availability processes that are
Poisson (and/or Gaussian) dominated with restricted asymmetry, i.e., $\max_{s,s'} |\sum_{j,k} \lambda_j r_{j,k,s}-\sum_{j,k} \lambda_j r_{j,k,s'}|$, 
being $O\left(\mbox{subpoly}(N)\right)$. For any $\alpha \in (0,1)$, $\exists N_\alpha$ s.t. $\forall N\ge N_\alpha$
in such systems of size $N$ that follow crowd scaling any arrival rate $\vlambda \in \alpha \mc{C}_{F,F}$ can be stabilized
by Alg.~\ref{alg:flexGreedy} and at the steady state (i.e., for any finite
$t$ when $t_0=-\infty$) the total number of unallocated steps 
in the system across all types is $O(\log N)$ w.p. $1-O\left(\frac{1}{N^2}\right)$.
\end{theorem}

\section{Flexible Agents and Inflexible Steps}
\label{sec:flexAgentsInflexSteps}
Now consider the setting where agents may split their available service-time across their skills, but a step must be allocated to one agent. 
Multiple agents cannot pool their service time to serve a step.  As before, for an agent availability vector $\bu$, there is a set of possible allocations of 
steps (of different $(j,k)$-types) to agents, denoted $C_{F,I}(\bu)$. Given $C_{F,I}(\bu)$ and the distribution of agent availability $\Gamma(\bu)$, we can 
define a capacity region $\mc{C}_{F,I}$ in the same way as $\mc{C}_{I,F}$ is defined in Sec.~\ref{sec:inflexAgentsFlexSteps} based on $C_{I,F}$.
Similarly, the generic centralized routine can be adapted by changing the optimization over $C(\mbf{U}(t))$ to an optimization over $C_{F,I}(\mbf{U}(t))$ 
while ensuring optimality of the modified scheme for $(F,I)$ system.

Allocation constraint \eqref{eq:taskAlloCons2} is for allocation of steps to a particular agent. For a given set of agents of different types
$\mbf{u}=(u_1, u_2, \dots, u_M)$ the allocation constraint can be written based on \eqref{eq:taskAlloCons2}. 
Note that for inflexible steps agents cannot pool service-times to serve a step. Consider a set of available agents $a_1, a_2, \dots, a_R$, of types $1, 2, \ldots,
m_r$ respectively. An allocation of $\{s_{j,k} \in \Z_+: k \in [K_j], j \in [N]\}$ steps to these agents is 
possible if and only if there are integers $\{z_{j,k,r} \in \Z_+\}$ such that $z_{j,k,r}$ $(j,k)$-steps are allocated to agent $a_r$ and all $\{s_{j,k}\}$ 
steps are allocated to some agent, i.e., for each $r$ there is an $\alpha_r$ 
in an $S$-dimensional simplex so that:
\begin{align}
& \sum_{j,k} \mbf{r}_{j,k} z_{j,k,r} \le \alpha_r h_{m_r}, \alpha_{r,s}=0 \ \mbox{if}
\ s \not\in S_{m_r}, r \in [R] \nonumber \\
& \sum_{r} z_{j,k,r} \ge s_{j,k} \forall j,k. \nonumber 
\end{align}
Hence, the optimization problem in the centralized allocation routine for $(F,I)$ system is an integer LP of the form:
\begin{align}
&\max_{s_{j,k} \in \Z_+} \sum_{j} \sum_{k=1}^{K_{j}} w_{j,k} s_{j,k} \nonumber \\
&\mbox{s.t. }  \sum_{j,k} \mbf{r}_{j,k} z_{j,k,r} \le \alpha_r h_{m_r}, \alpha_{r,s}=0 \
 \mbox{if} \ s \not\in S_{m_r}, r \in [R]\mbox{,} \nonumber \\
&\quad \sum_{r} z_{j,k,r} \ge s_{j,k} \mbox{ for all } j,k \mbox{, and} \nonumber  \\
&\quad z_{j,k,r} \in \Z_+ \mbox{ for all } j,k,r, R=\sum_m u_m \mbox{.} \label{eq:centAlgoFI1}
\end{align}
Note that like \eqref{eq:centAlgo}, the objective can be written as $\sum_{j} \sum_{k=1}^{K_{j-1}} v_{j,k} s_{j,k}$ where $v_{j,k}=\max(w_{j,k},0)$.

This problem has a special structure which leads to a computationally-efficient algorithm. Consider the following.
\begin{align}
&\max_{z_{j,k,r} \in \Z_+} \sum_{j} \sum_{k=1}^{K_{j}} v_{j,k} \sum_r z_{j,k,r} \nonumber \\
&\mbox{s.t. } \sum_{j,k} \mbf{r}_{j,k} z_{j,k,r} \le \alpha_r h_{m_r}, \alpha_{r,s}=0 \
 \mbox{if} \ s \not\in S_{m_r}, r \in [R] \mbox{, and} \nonumber \\
&\quad z_{j,k,r} \in \Z_+ \mbox{ for all } j,k,r, R=\sum_m u_m \mbox{.} \label{eq:centAlgoFI2}
\end{align}

When operating at the optimum of \eqref{eq:centAlgoFI1}, $\sum_{r} z_{j,k,r} = s_{j,k}$, 
and so we see that \eqref{eq:centAlgoFI1} and \eqref{eq:centAlgoFI2} 
have the same optimal value. Hence, we solve problem \eqref{eq:centAlgoFI2} 
instead of problem \eqref{eq:centAlgoFI1}.

Note that as there is no constraint between $\{z_{j,k,r}:j,k\}$ and $\{z_{j,k,r'}: j,k\}$, problem \eqref{eq:centAlgoFI2} decomposes
into $\sum_m u_m$ optimization problems, each for an available agent. Consider the optimization problem for an agent of type $m$.
\begin{align}
&\max_{z_{k,j}} \sum_j \sum_{k=1}^{K_j} v_{j,k} z_{j,k} \nonumber \\
&\mbox{s.t. } \sum_{j,k} \mbf{r}_{j,k} z_{j,k} \le \alpha_r h_{m_r}, \alpha_{r,s}=0
\ \mbox{if} \ s \not\in S_{m_r} \mbox{,} 
\end{align}
which is again equivalent to the following problem, expressed in terms of the set $S_m$ of skills of type $m$ agent:
\begin{align}
&\max_{z_{k,j}} \sum_{j,k:r_{j,k,s}=0 \ \mbox{if} \ s \not\in S_m} v_{j,k} z_{j,k} \nonumber \\
&\mbox{s.t. } \sum_{j,k:r_{j,k,s}=0 \ \mbox{if} \ s \not\in S_m} z_{j,k} \sum_{s} r_{j,k,s} \le h_m \mbox{.}
\label{eq:centAlgoFI3}
\end{align}
This is a one-dimensional knapsack problem, and there are dynamic programming (DP) pseudo-polynomial time algorithms for solving. 
Since the sack size $h_m$ is finite (does not scale with the system), the DP has computational complexity $O(\sum_j K_j)$.
This implies the centralized scheme decomposes into $\sum_m u_m$ problems, each of which can be solved in polynomial time.

Thus, the centralized scheme naturally leads to a decentralized scheme where each available agent solves \eqref{eq:centAlgoFI3} and uses
the optimal solution as its \emph{potential} allocated steps. Agents may use an arbitrary contention mechanism among themselves to decide 
which agent allocates first.  Upon resolving contention, agents pick steps greedily by solving \eqref{eq:centAlgoFI3}. Since the decentralized scheme follows 
directly from the centralized one \eqref{eq:centAlgoFI1}, performance guarantees from Thm.~\ref{thm:centralized} hold.

Although this simple decentralized scheme is optimal, it does not give customers freedom of choosing agents.  Thus, we propose
another decentralized scheme where customers get to pick any agent from a subset of available agents.

\begin{algorithm}
\caption{Restricted Greedy}
\label{alg:restGreedy}
{\em Compute and Store: one-time}
\begin{algorithmic}[1]
\FOR{$d=1:D$}
\STATE $\mc{A}_d=\cup_j \{\mbox{set of required skills for steps at depth} \ $d$ \ \mbox{of} \ T_j\}$ 
\STATE $p_d=0$
\WHILE{$\mc{A}_d \neq \emptyset$}
\STATE $p_d=p_d+1$
\STATE Pick maximal subsets from the collection
$\mc{A}_d$, say $\mc{L}^d_{p_d}$.
\STATE $\mc{A}_d = \mc{A}_d \backslash \mc{L}^d_{p_d}$
\ENDWHILE
\ENDFOR
\end{algorithmic}

{\em Allocation at time $t$}

Define $D=\max_j \ \mbox{depth of} \ T_j$

\begin{algorithmic}[1]
\FOR{$d=1:D$}
\WHILE{$k \le p_d$}
\STATE Steps in $\mc{L}^d_k$ allocate themselves greedily (ties are broken
arbitrarily)
\STATE $k=k+1$
\ENDWHILE
\ENDFOR
\end{algorithmic}
\end{algorithm}
Alg.~\ref{alg:restGreedy} allows the different types of steps to pick agents greedily, but in a restricted
manner. It prioritizes steps with lower depth like Prioritized Greedy.
Among steps with the same priority (in terms of depth), it gives preference
to steps requiring more skills to ensure an agent with multiple 
skills is not used unwisely for a step with
lesser requirements. 
\begin{theorem}
\label{thm:restGreedy}
Consider a flexible agent and inflexible steps crowdsourcing system where each type of agent has
$O(1)$ skills, and arrival as well as availability processes are Poisson (and/or Gaussian) dominated, $\{S_m:m \in [M]\}$ is
a partition of $[S]$ and $\sum_s r_{j,k,s}$ are same for all $(j,k)$. 
For this system for any $\alpha \in (0,1)$, $\exists N_\alpha$ s.t. $\forall N\ge N_\alpha$
in such systems of size $N$ that follows crowd scaling any arrival rate $\vlambda \in \alpha \mc{C}_{F,I} $ can be stabilized
by Alg.~\ref{alg:restGreedy} and at steady-state the total number of unallocated steps 
in the system across all types is $O(\log N)$ w.p.\ $1-O\left(\frac{1}{N^2}\right)$.
\end{theorem}

For many systems the total sizes of steps are nearly identical and so 
the assumption on total size is not restrictive, though results can be 
extended to the case where the total sizes are random with the same mean. 
The assumption $\{S_m: m \in [M]\}$ is a partition is required for proving
the performance guarantee, but the algorithm (actually a simpler version)
works well on simulations. The above performance guarantee
can be extended for the following conditions.
$\{S_{i}: i \in \mc{I}\}$ is a partition of $[S]$ 
for some $\mc{I} \subset [M]$ and for any $m$~~$S_m \subset S_{i}$ for 
some $i \in \mc{I}$, $D=1$, and for any $(j,k)$, $(j',k')$ pair, 
$\{s: r_{j,k,s}>0\}$ and $\{s: r_{j',k',s}>0\}$ either have no 
intersection or one is a subset of the other.

\section{Evaluation}
\label{sec:evaluations}

Secs.~\ref{sec:capacity}--\ref{sec:flexAgentsInflexSteps} characterized 
limits of different types of crowdsourcing systems, proposed efficient policies 
for optimal centralized allocation and designed decentralized schemes with
provable bounds on backlog while giving customers freedom of choice. This 
section complements theoretical results by studying 
real data from Samasource, a non-profit crowdsourcing 
company and
realistic Monte Carlo simulations. We study performance of simplified (in implementation and computation) 
versions of proposed decentralized algorithms above.

Let us first describe evaluating allocation using real data. The dataset 
contains $9.3$M tasks and each belongs to a specific project.
Some projects are regarded as \emph{real-time} which means they have 
higher priority. The overall number of tasks that belong to the real-time projects 
is about $4.2$M. Each task comprises 1 or 2 steps which in turn comprises a single substep. 
Some tasks have strict step ordering, i.e., the previous 
step must be completed before the next could be scheduled. Average substep working 
time requirement is $340$ sec. From the data, we calculate the 
turn-around time (TAT) for each task, i.e., the time since the task arrived to 
the system until the time its last step was completed. The cdfs of TAT for all 
projects and for real-time projects only are given in Fig.~\ref{fig:sama_cdf_3_days_tat}. 

SamaHub, the platform of Samasource considers both agents and steps to be flexible. 
We implement a simplified version of the relevant decentralized algorithm,
called $\algo1$, where we prioritize the steps with higher precedence to choose agents greedily with random tie-breaking. 


\begin{figure}
	\centering
	\includegraphics[width=3.5in]{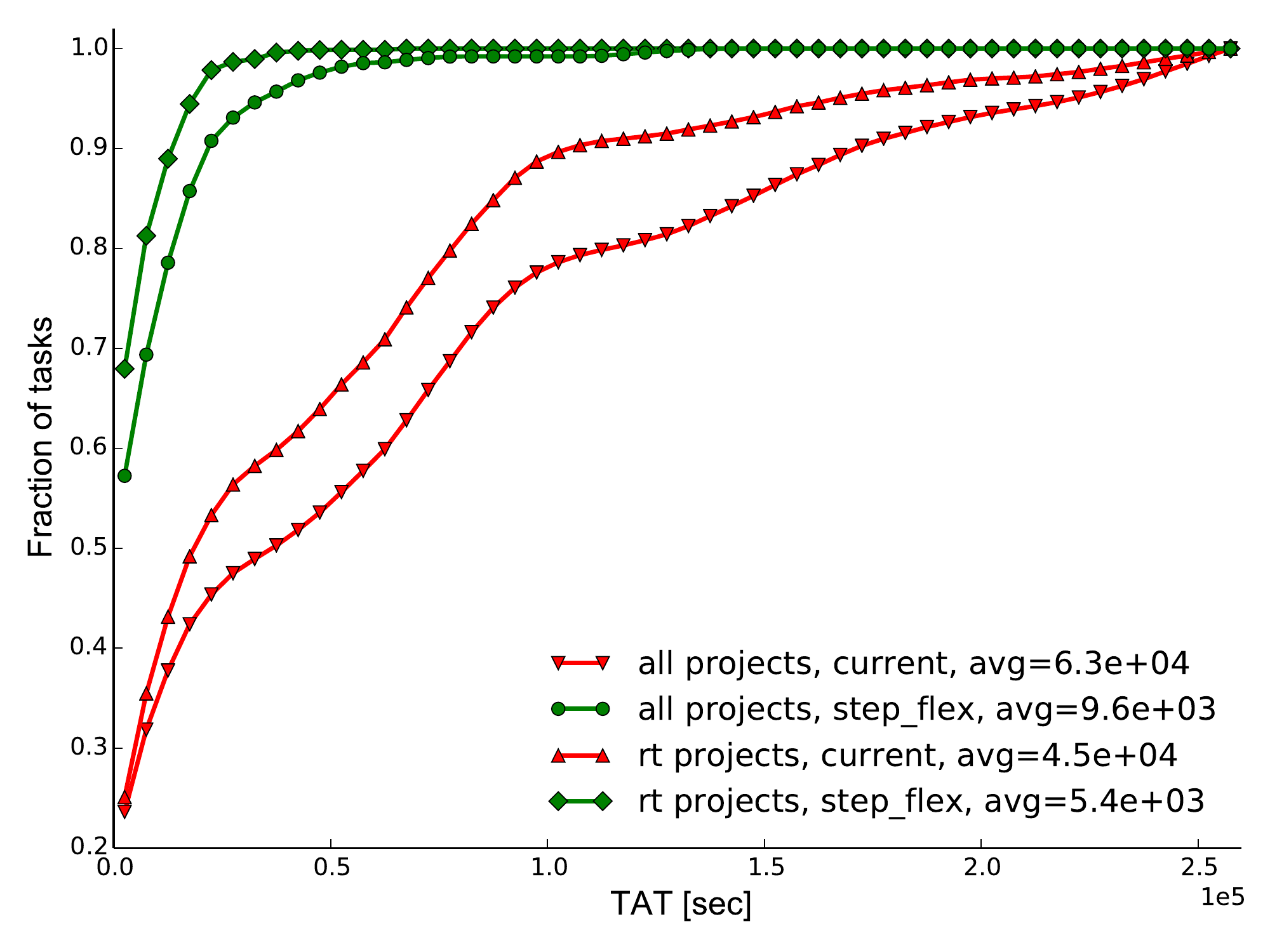}
    \caption{\small CDF of tasks turn-around time (TAT) using real dataset. Current allocation on the platform ``current'' vs our algorithm ``$\algo1$".}
    \label{fig:sama_cdf_3_days_tat}
\end{figure}


\begin{figure}
	\centering
	\subfloat[]{\includegraphics[width=0.48\columnwidth]{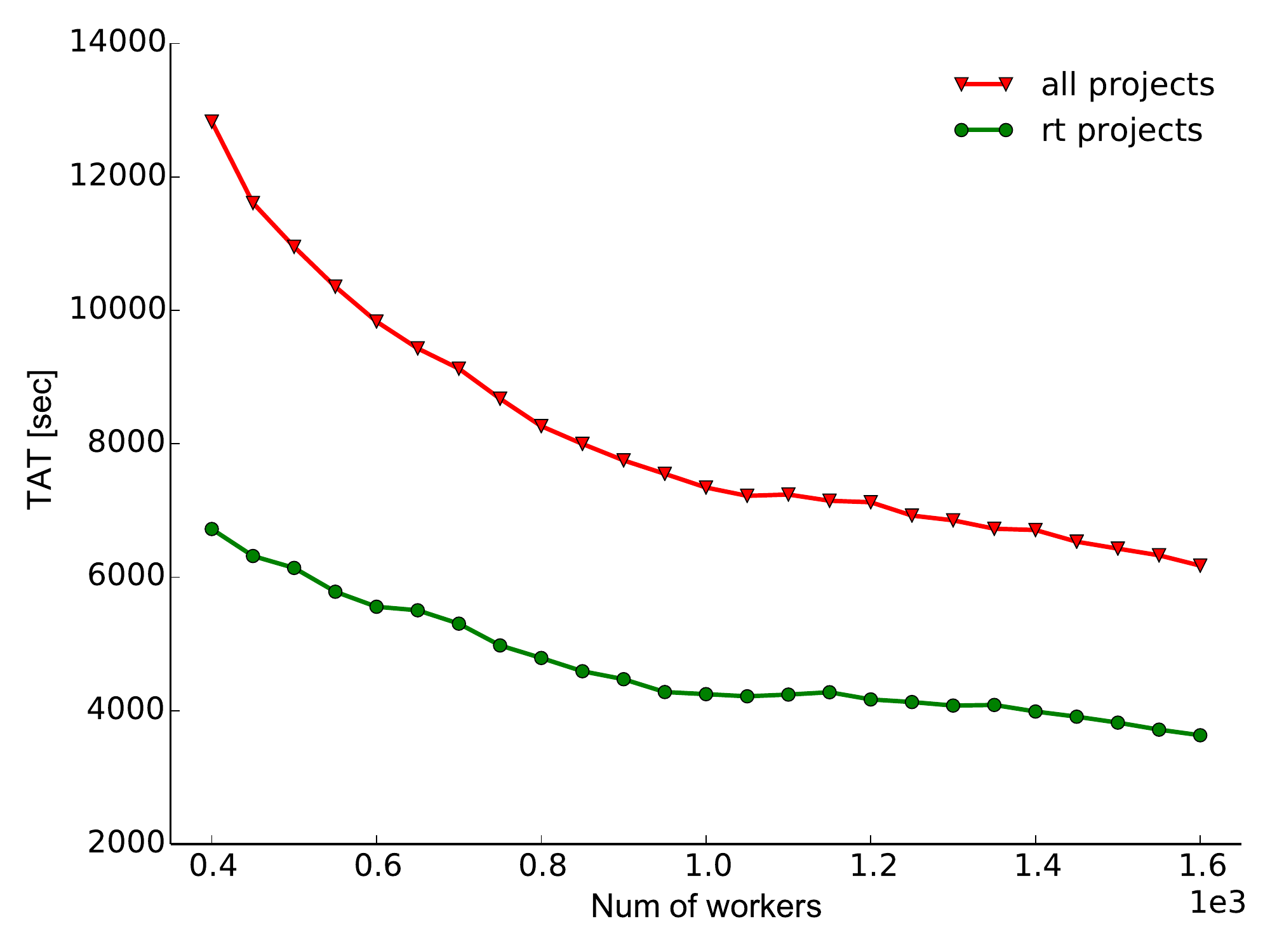}}  $\;$ 
	\subfloat[]{\includegraphics[width=0.48\columnwidth]{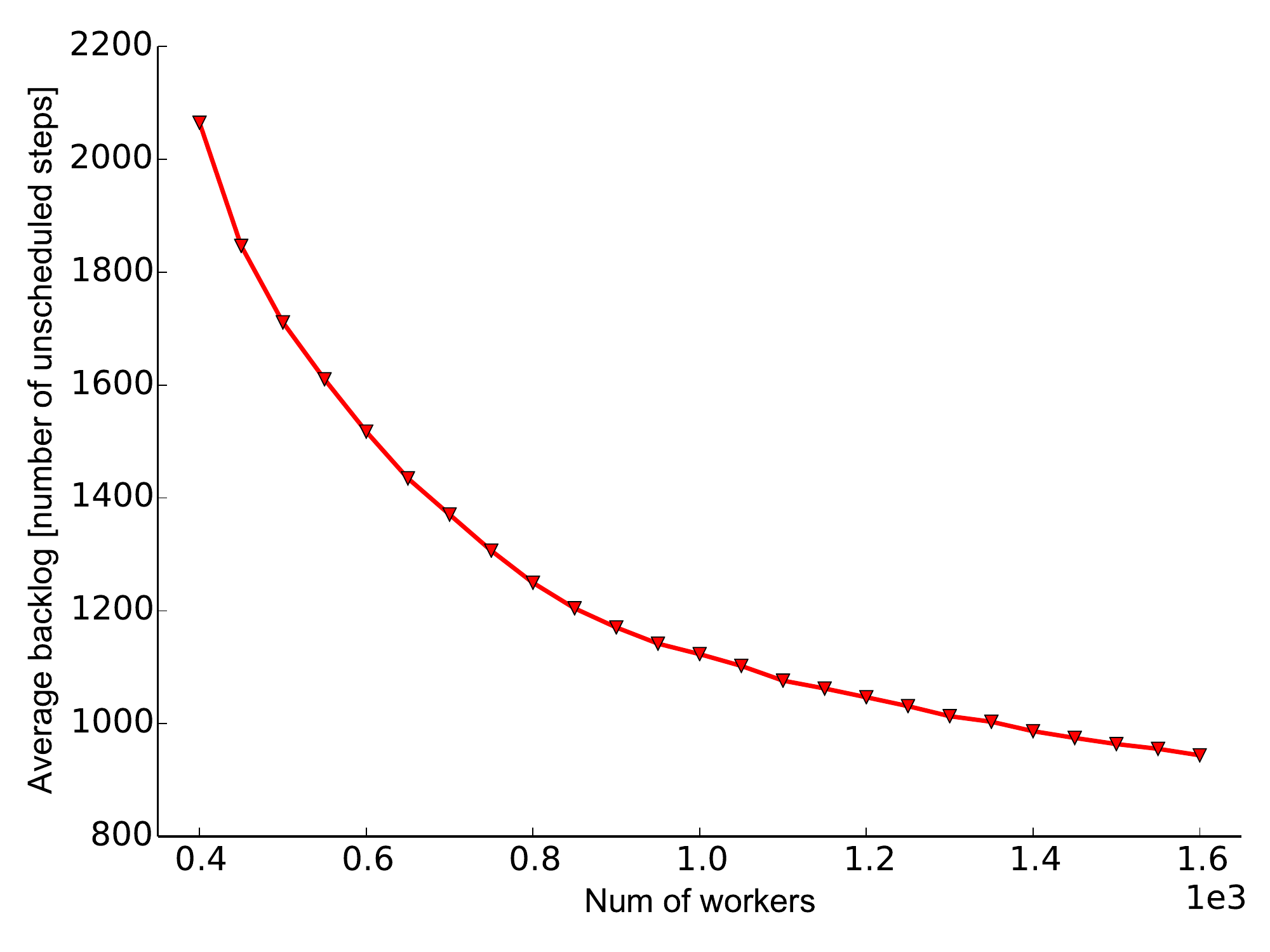}} $\;$
       \caption{\small Performance of our $\algo1$ algorithm on real data, as a function of number of workers. (a) Tasks turn-around time (TAT). (b) Average backlog (number of unallocated steps in the system). }
    \label{fig:sama_workers_algo1}
\end{figure}

To compare current allocation on SamaHub with our approach, we use real 
data as input to $\algo1$. Since we lack exact 
knowledge of worker availability, we make the following assumption in consultation with Samasource. The number of 
active workers in the system is $625$, evenly distributed across four time zones: $-4,0,3,5.5$, where each worker works every day from $9$am to $5$pm. 
Each worker possesses the skills required for any substep in the dataset. Fig.~\ref{fig:sama_cdf_3_days_tat} compares 
the cdf of TAT of our approach $\algo1$ (simulated with the data as input) with currently deployed scheme. Our algorithm substantially outperforms current 
scheme: average TAT for all projects is $\times 6.5$ better and 
more than $\times 8$ better for real-time projects. This improvement is also influenced by our implementation, which is not restricted by the currently-practiced organizational structure.

Fig.~\ref{fig:sama_workers_algo1} shows how $\algo1$ performs as a function of number of workers. As the number of workers grows, 
TAT decreases (see Fig.~\ref{fig:sama_workers_algo1}(a)).  The benefit of adding more workers can be seen even more clearly when 
analyzing backlog, i.e., the average number of steps that entered the system but not yet scheduled, see Fig.~\ref{fig:sama_workers_algo1}(b).


\begin{figure*}[]
	\centering
    \subfloat[]{\includegraphics[width=0.55\columnwidth]{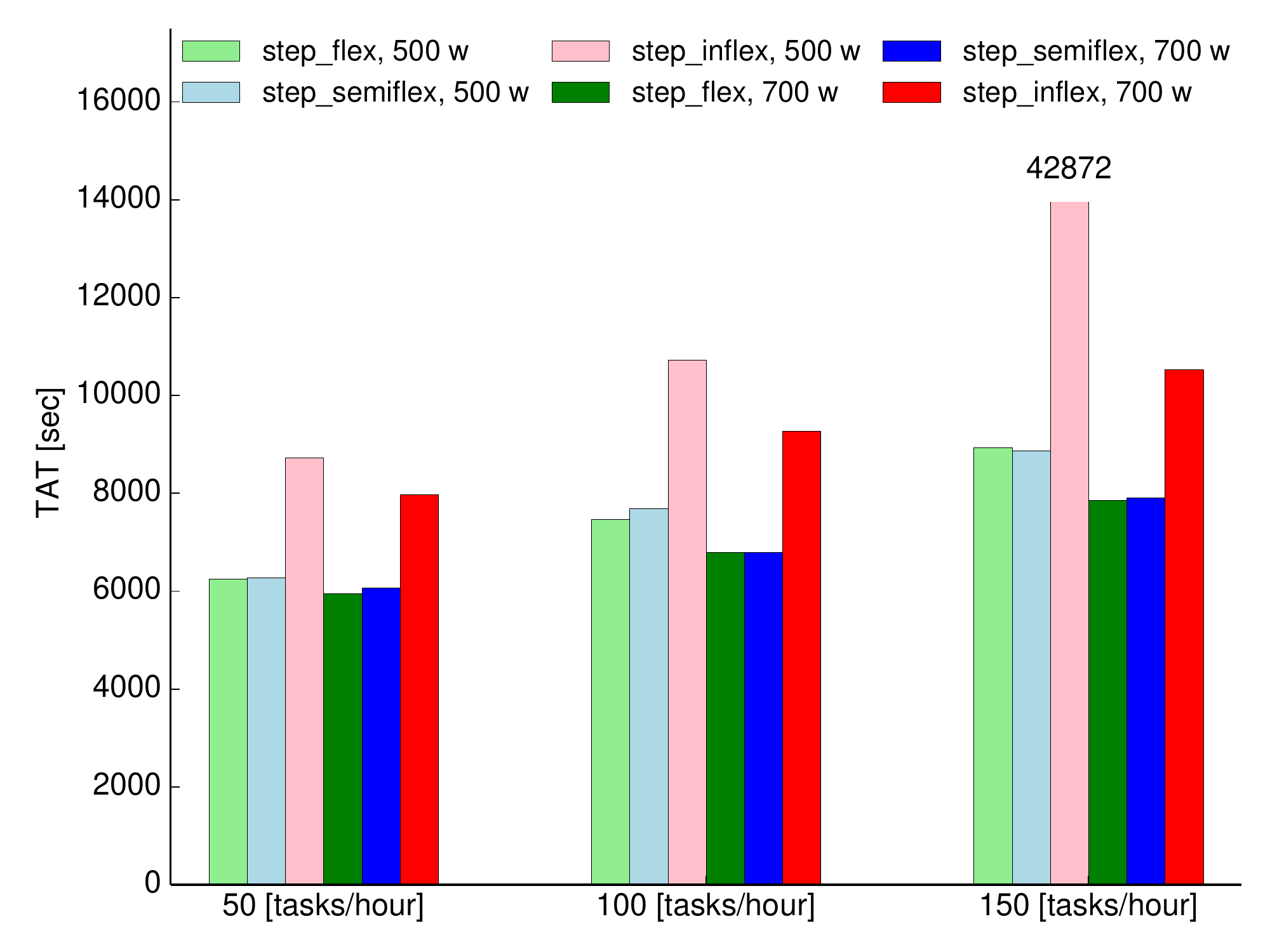}}  $\;$ 
	\subfloat[]{\includegraphics[width=0.55\columnwidth]{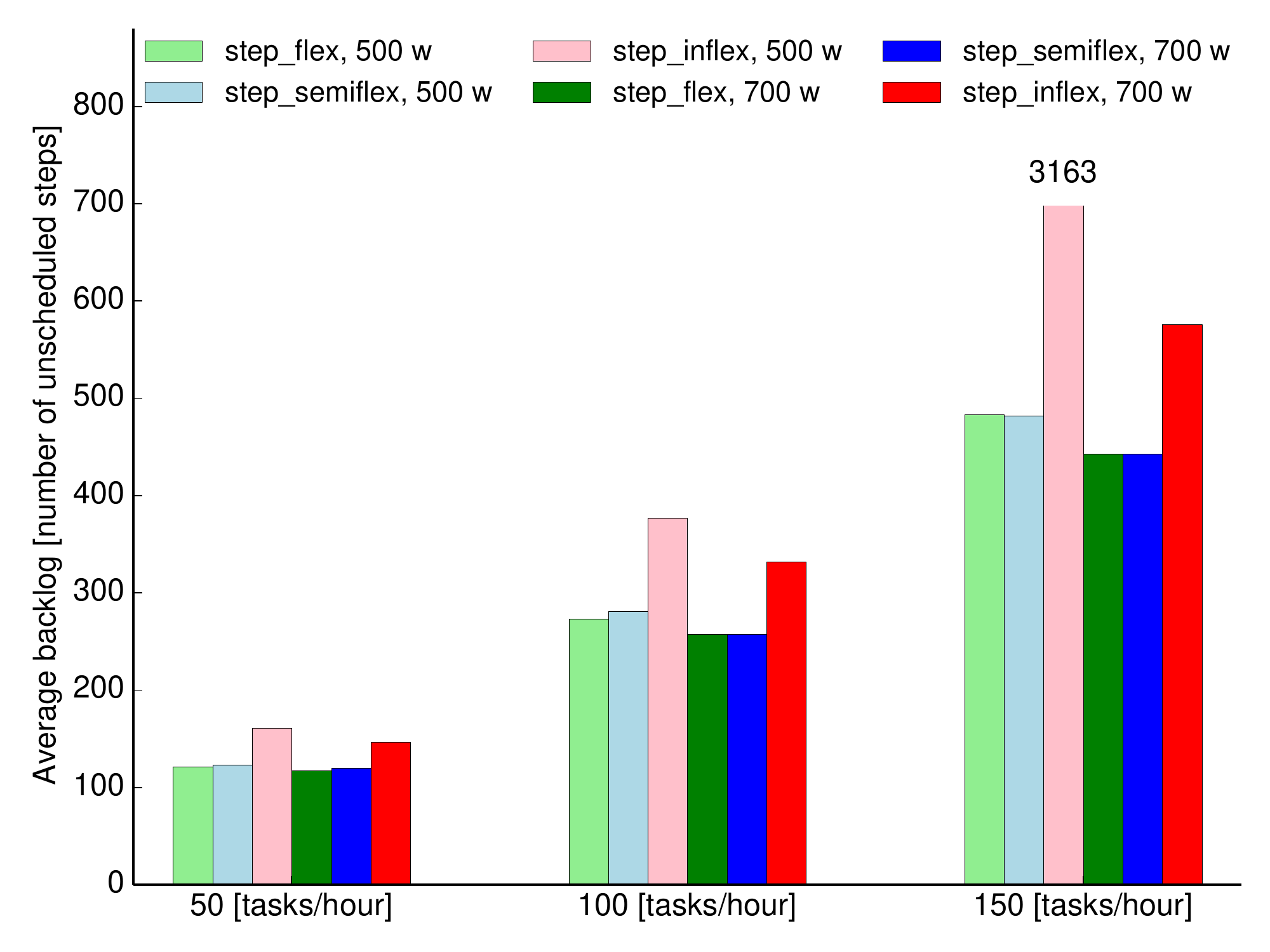}} $\;$
    \subfloat[]{\includegraphics[width=0.55\columnwidth]{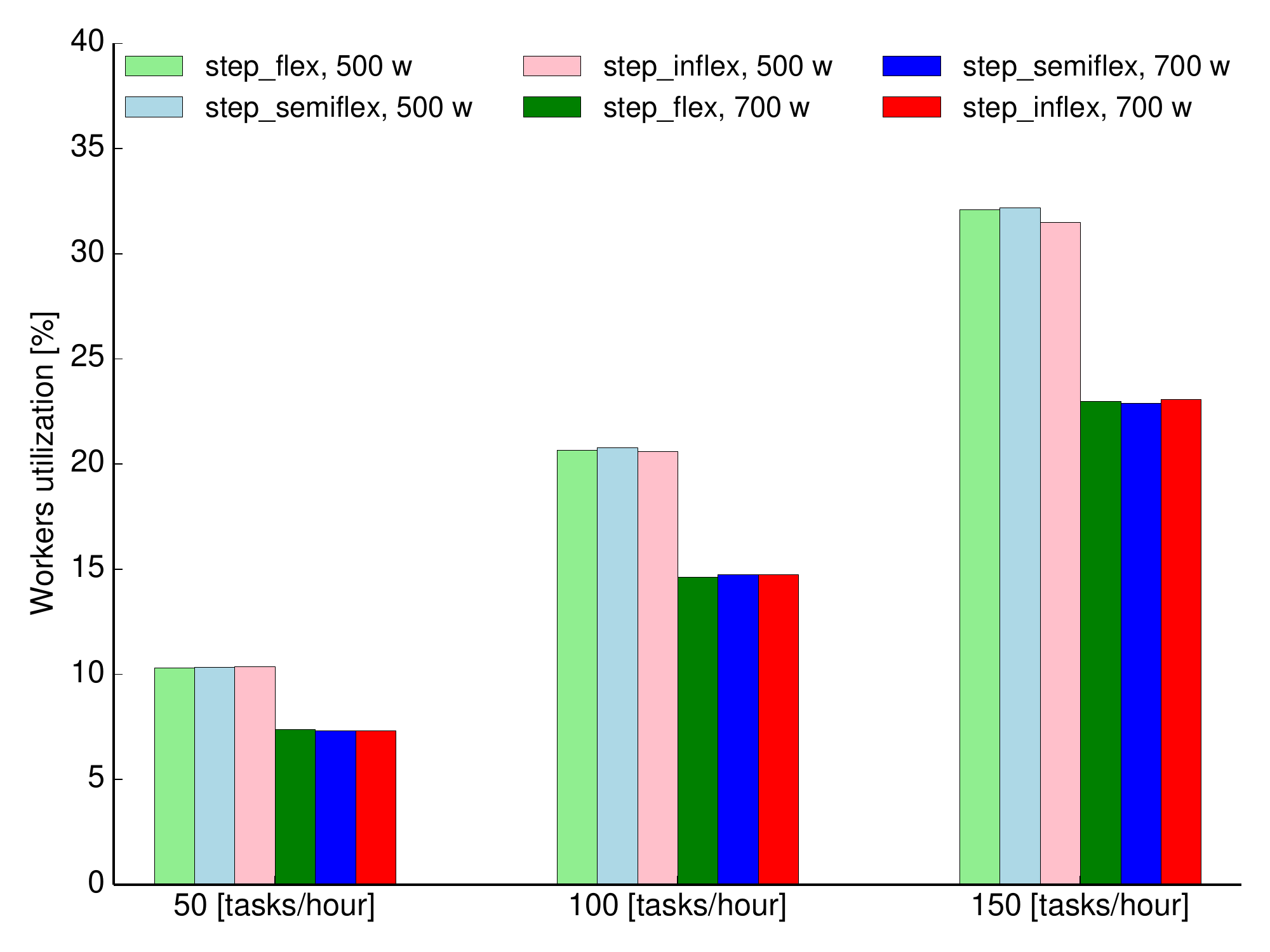}} $\;$
    \caption{\small Performance of our algorithms on synthetic data with \emph{short} sub-steps ($60-600$ sec), as a function of load. (a) Tasks turn-around time (TAT). (b) Average backlog (number of unscheduled steps in the system). (c) Workers utilization.}
    \label{fig:synth_load_bars}
\end{figure*}

\begin{figure*}[]
	\centering
    \subfloat[]{\includegraphics[width=0.55\columnwidth]{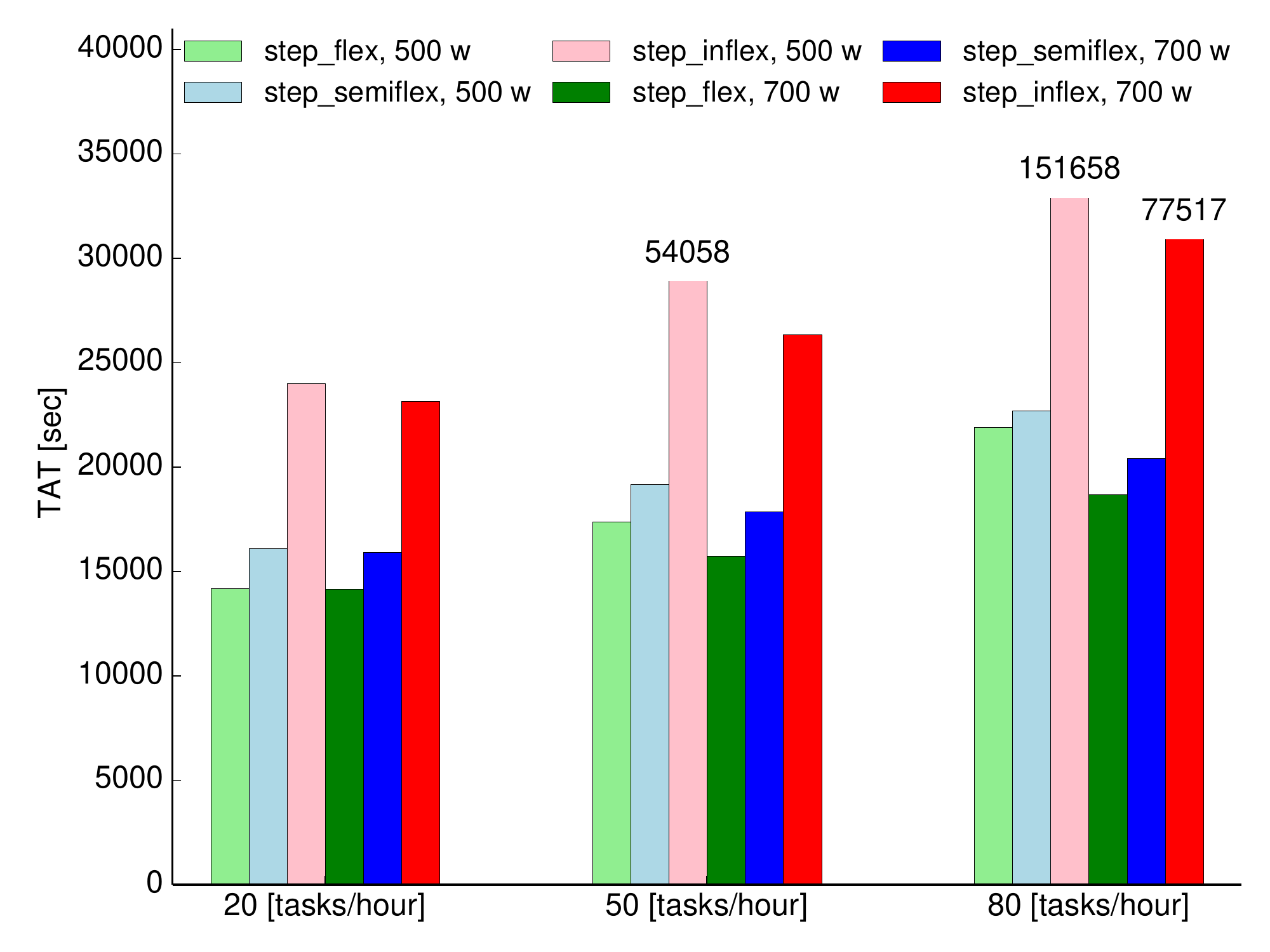}}  $\;$ 
	\subfloat[]{\includegraphics[width=0.55\columnwidth]{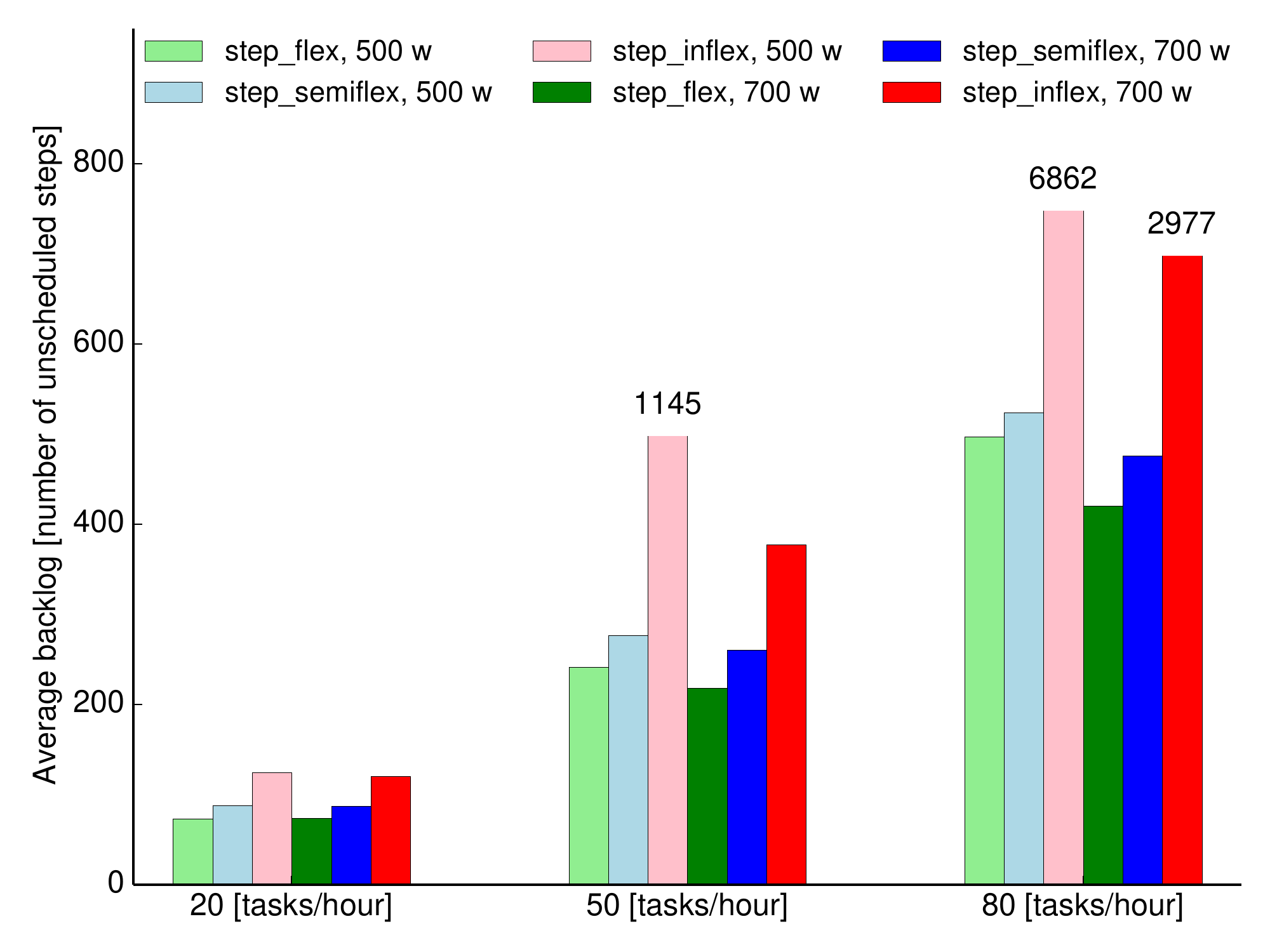}} $\;$
    \subfloat[]{\includegraphics[width=0.55\columnwidth]{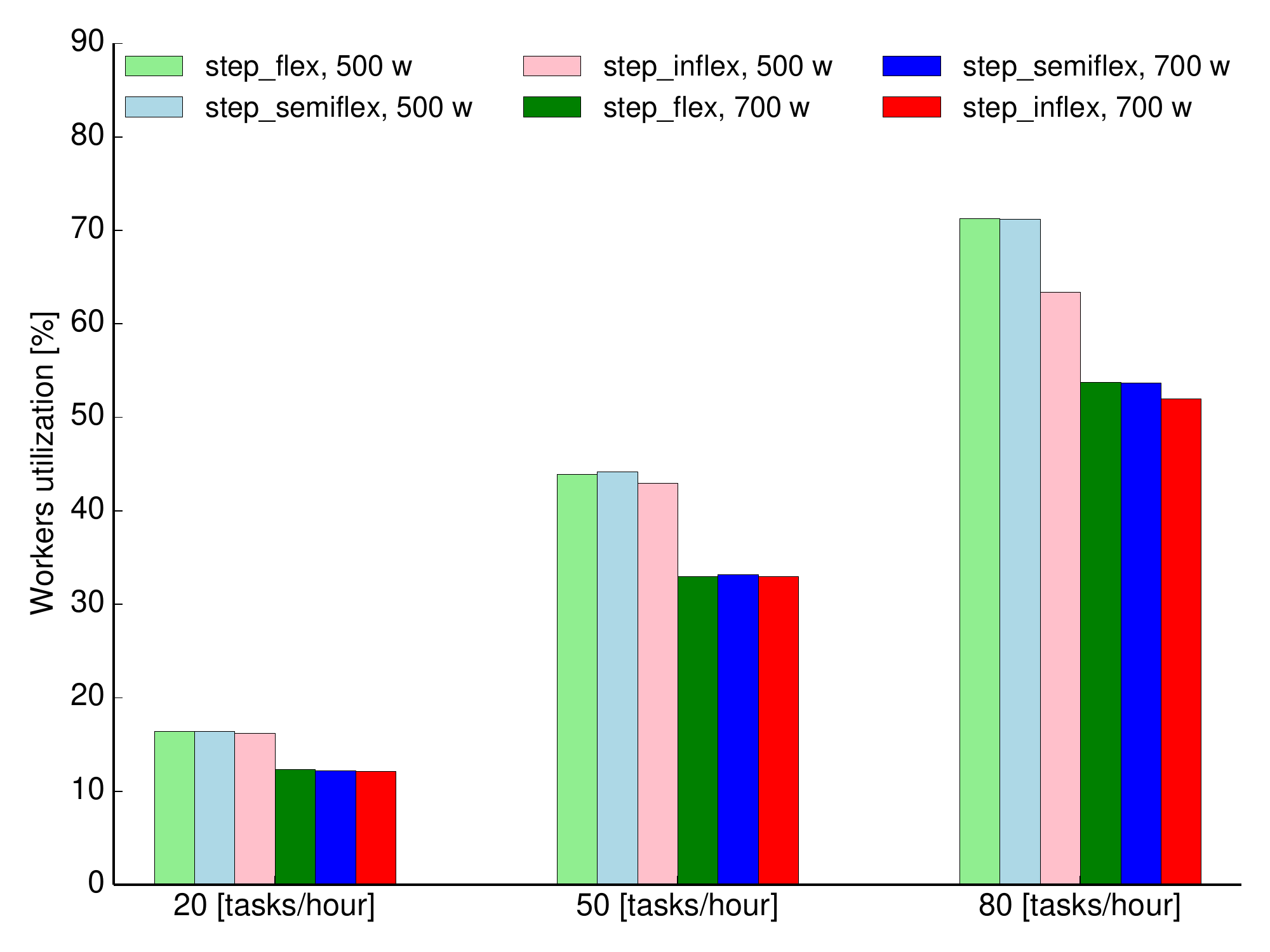}} $\;$
    \caption{\small Performance of our algorithms on synthetic data with \emph{long} sub-steps ($600-6000$ sec), as a function of load. (a) Tasks turn-around time (TAT). (b) Average backlog (number of unscheduled steps in the system). (c) Workers utilization.}
    \label{fig:synth_load_bars_long_skills}
\end{figure*}


We also evaluate our algorithms on synthetic data, considering flexible agents and flexible
steps, and flexible agents and inflexible steps. Algorithm $\algo1$ is used for the first system and a simplified
version of the Restricted Greedy scheme, $\algo3$, where we prioritize steps with
higher skill requirements and allocate among them greedily is used for the second.
We also consider a scenario in between flexible and
inflexible steps, where each substep is allocated to a single agent,
but different substeps of a step can be allocated to different agents. For
this, we develop $\algo2$ where steps allocate themselves greedily
while ensuring a substep gets all service from an agent.
We expected $\algo1$ to outperform $\algo3$, but we found somewhat surprisingly  that $\algo1$ and 
$\algo2$ perform very similarly.

The first set of generated data has tasks with up to three steps in each and with strict 
ordering. Each step comprises one to three random substeps out of five 
possible types. Working time requirement for each substep is uniformly distributed 
between $60$ and $600$ sec. Each worker in the system has daily availability 
from $9$am to $5$pm, evenly distributed across four time zones: 
$-4,0,3,5.5$. A worker possesses a random set of skills that enables her to work 
on up to three (out of five) substep types. For each of our three algorithms we 
compare three metrics: TAT, backlog queue, and worker utilization. The experiment 
simulated a single run over a timespan of $40$ days.

Fig.~\ref{fig:synth_load_bars} shows algorithms $\algo1$ and 
$\algo2$ outperform $\algo3$ for both cases: $500$ workers in the system and 
$700$ workers. When the load on the system is $150$ tasks/hour and the number 
of workers is $500$, algorithm $\algo3$ is substantially worse 
since it becomes unstable for this load. Notice that $\algo1$ and $\algo2$ perform
very similarly, which can be explained by relatively short substep work time 
requirement (in which case splitting becomes a rare event). 
Also note that worker utilization of $\algo3$ is not much worse than of the other algorithms. This can be explained by the 
long backlog queue of $\algo3$. Though it is harder for $\algo3$ to find a 
worker capable of working on the whole step, when the backlog becomes 
large, the probability that a given worker will be assigned to some whole 
step grows.

The last set of results uses the same synthetic data as before, but the working time requirement for each substep is 
now uniformly distributed between $600$ and $6000$ sec.
Fig.~\ref{fig:synth_load_bars_long_skills} shows a slight advantage of $\algo1$ over $\algo2$. Due 
to the longer working time requirements per substep, cases in which a substep may be split to improve allocation are more 
probable. In this scenario, the disadvantage of $\algo3$ is more obvious: 
for a load of $50$ tasks/hour and $1200$ workers, its TAT and backlog are very large and unstable.

To summarize, our approach substantially outperforms Samasource's current allocation 
scheme. While $\algo1$ achieves  
best performance in terms of TAT and backlog, $\algo2$ may be a good 
alternative. Its performance is almost the same but does not require 
splitting substeps among different workers, and is 
computationally lighter.

\section{Conclusion}
\label{sec:conc}

Inspired by skilled crowdsourcing systems, we have developed new algorithms for allocating tasks to agents while handling novel system properties
such as vector-valued service requirements, precedence and
flexibility constraints, random and time-varying resource availability, 
large system size, need for simple decentralized 
schemes requiring minimal actions from the platform provider, and the freedom of customers to choose agents without compromising
system performance.  We have provided capacity regions, asymptotic performance guarantees for decentralized algorithms, and demonstration
of efficacy in practical regimes, via large-scale data from a non-profit 
crowdsourcing company.

\bibliographystyle{IEEEtran}
\bibliography{abrv,conf_abrv,crowd}

\begin{thebibliography}{10}
\providecommand{\url}[1]{#1}
\csname url@samestyle\endcsname
\providecommand{\newblock}{\relax}
\providecommand{\bibinfo}[2]{#2}
\providecommand{\BIBentrySTDinterwordspacing}{\spaceskip=0pt\relax}
\providecommand{\BIBentryALTinterwordstretchfactor}{4}
\providecommand{\BIBentryALTinterwordspacing}{\spaceskip=\fontdimen2\font plus
\BIBentryALTinterwordstretchfactor\fontdimen3\font minus
  \fontdimen4\font\relax}
\providecommand{\BIBforeignlanguage}[2]{{%
\expandafter\ifx\csname l@#1\endcsname\relax
\typeout{** WARNING: IEEEtran.bst: No hyphenation pattern has been}%
\typeout{** loaded for the language `#1'. Using the pattern for}%
\typeout{** the default language instead.}%
\else
\language=\csname l@#1\endcsname
\fi
#2}}
\providecommand{\BIBdecl}{\relax}
\BIBdecl

\bibitem{ChatterjeeBVV2016}
A.~Chatterjee, M.~Borokhovich, L.~R. Varshney, and S.~Vishwananth, ``Efficient
  and flexible crowdsourcing of specialized tasks with precedence
  constraints,'' in \emph{Proc. 2016 IEEE INFOCOM}, Apr. 2016 (to appear).

\bibitem{Cuenin2015}
\BIBentryALTinterwordspacing
A.~Cuenin, ``Each of the top $25$ best global brands has used crowdsourcing,''
  Jun. 2015. [Online]. Available:
  \url{http://www.crowdsourcing.org/editorial/each-of-the-top-25-best-global-brands-has-used-crowdsourcing/50145}
\BIBentrySTDinterwordspacing

\bibitem{TapscottW2006}
D.~Tapscott and A.~D. Williams, \emph{Wikinomics: How Mass Collaboration
  Changes Everything}, expanded~ed.\hskip 1em plus 0.5em minus 0.4em\relax New
  York: Portfolio Penguin, 2006.

\bibitem{GinoS2012}
F.~Gino and B.~R. Staats, ``The microwork solution,'' \emph{Harvard Bus. Rev.},
  vol.~90, no.~12, pp. 92--96, Dec. 2012.

\bibitem{MarcusP2015}
A.~Marcus and A.~Parameswaran, ``Crowdsourced data management: Industry and
  academic perspectives,'' \emph{Foundations and Trends in Databases}, vol.~6,
  no. 1-2, pp. 1--161, Dec. 2015.

\bibitem{MaloneLD2010}
T.~W. Malone, R.~Laubacher, and C.~Dellarocas, ``The collective intelligence
  genome,'' \emph{MIT Sloan Manage. Rev.}, vol.~51, no.~3, pp. 21--31, Spring
  2010.

\bibitem{BoudreauL2013}
K.~J. Boudreau and K.~R. Lakhani, ``Using the crowd as an innovation partner,''
  \emph{Harvard Bus. Rev.}, vol.~91, no.~4, pp. 60--69, Apr. 2013.

\bibitem{DustdarG2011}
S.~Dustdar and M.~Gaedke, ``The social routing principle,'' \emph{{IEEE}
  Internet Comput.}, vol.~15, no.~4, pp. 80--83, July-Aug. 2011.

\bibitem{DiPalantinoV2009}
D.~DiPalantino and M.~Vojnovi\'{c}, ``Crowdsourcing and all-pay auctions,'' in
  \emph{Proc. 10th ACM Conf. Electron. Commer. (EC'09)}, Jul. 2009, pp.
  119--128.

\bibitem{VarshneyACSOLR2014_arXiv}
L.~R. Varshney, S.~Agarwal, Y.-M. Chee, R.~R. Sindhgatta, D.~V. Oppenheim,
  J.~Lee, and K.~Ratakonda, ``Cognitive coordination of global service
  delivery,'' arXiv:1406.0215v1 [cs.OH]., Jun. 2014.

\bibitem{KleinbergT2006}
J.~Kleinberg and {\'{E}}.~Tardos, \emph{Algorithm Design}.\hskip 1em plus 0.5em
  minus 0.4em\relax Addison-Wesley, 2005.

\bibitem{SrikantY2014}
R.~Srikant and L.~Ying, \emph{Communication Networks: An Optimization, Control
  and Stochastic Networks Perspective}.\hskip 1em plus 0.5em minus 0.4em\relax
  Cambridge University Press, 2014.

\bibitem{Pinedo2012}
M.~L. Pinedo, \emph{Scheduling: Theory, Algorithms, and Systems}.\hskip 1em
  plus 0.5em minus 0.4em\relax Springer, 2012.

\bibitem{pang2014service}
G.~Pang and A.~L. Stolyar, ``A service system with on-demand agent
  invitations,'' \emph{Queueing Systems}, Nov. 2015.

\bibitem{ChatterjeeVV2015}
A.~Chatterjee, L.~R. Varshney, and S.~Vishwananth, ``Work capacity of freelance
  markets: Fundamental limits and decentralized schemes,'' in \emph{Proc. 2015
  IEEE INFOCOM}, Apr. 2015, pp. 1769--1777.

\bibitem{ChudakS1999}
F.~A. Chudak and D.~B. Shmoys, ``Approximation algorithms for
  precedence-constrained scheduling problems on parallel machines that run at
  different speeds,'' \emph{J. Algorithms}, vol.~30, no.~2, pp. 323--343, Feb.
  1999.

\bibitem{Pedarsani2015}
R.~Pedarsani, ``Robust scheduling for queueing networks,'' Ph.D. dissertation,
  University of California, Berkeley, Berkeley, CA, 2015.

\bibitem{TassiulasE1992}
L.~Tassiulas and A.~Ephremides, ``Stability properties of constrained queueing
  systems and scheduling policies for maximum throughput in multihop radio
  networks,'' \emph{{IEEE} Trans. Autom. Control}, vol.~37, no.~12, pp.
  1936--1948, Dec. 1992.

\bibitem{Neely2010}
M.~J. Neely, \emph{Stochastic Network Optimization with Application to
  Communication and Queueing Systems}.\hskip 1em plus 0.5em minus 0.4em\relax
  Morgan \& Claypool, 2010.

\bibitem{KellererPP2004}
H.~Kellerer, U.~Pferschy, and D.~Pisinger, \emph{Knapsack Problems}.\hskip 1em
  plus 0.5em minus 0.4em\relax Springer, 2004.

\bibitem{Pritchard2009}
D.~A.~G. Pritchard, ``Linear programming tools and approximation algorithms for
  combinatorial optimization,'' Ph.D. dissertation, University of Waterloo,
  2009.

\bibitem{BubeckC2012}
S.~Bubeck and N.~Cesa-Bianchi, ``Regret analysis of stochastic and
  nonstochastic multi-armed bandit problems,'' \emph{Found. Trends Mach.
  Learn.}, vol.~5, no.~1, pp. 1--122, Dec. 2012.

\bibitem{VukovicS2012}
M.~Vukovic and O.~Stewart, ``Collective intelligence applications in {IT}
  services business,'' in \emph{Proc. IEEE 9th Int. Conf. Services Comput.
  (SCC)}, Jun. 2012, pp. 486--493.

\bibitem{ChatterjeeVV2015_arXiv}
A.~Chatterjee, L.~R. Varshney, and S.~Vishwananth, ``Work capacity of freelance
  markets: Fundamental limits and decentralized schemes,'' arXiv:1508.00023
  [cs.MA], Jul. 2015.

\bibitem{Borkar2008}
V.~S. Borkar, \emph{Stochastic Approximation: A Dynamical Systems
  Viewpoint}.\hskip 1em plus 0.5em minus 0.4em\relax Cambridge University
  Press, 2008.

\end{thebibliography}

\newpage

\begin{appendices}
\section{Proofs}
\label{sec:proofs}
In this section we present proofs of the main results.

\subsection{Proof of Theorem \ref{thm:capacity}}
Here we only prove that any $\vlambda$ outside the closure of
$\mc{C}$ cannot be stabilized by any policy. To prove achievability, 
it is sufficient to show that there exists a policy that stabilizes any
$\vlambda$ in the interior of $\mc{C}$. Hence, it is sufficient to prove
Thm.~\ref{thm:centralized}, which we do later.

This proof consists of the following steps. We first compare two systems,
the original system in question and another in which there is no precedence 
constraint among different steps of a job. We claim that on any sample path
under any policy for the first system, there exists a policy in the second
system so that the total number of incomplete jobs across all job types
in the second is a lower bound (sample path-wise) for that in the first. Then 
we show that the second system cannot be stabilized for a $\vlambda$ outside 
the closure of $\mc{C}$ and so the result follows for the first system.

Note that the claim regarding the number of incomplete jobs across all types
in the second system being a lower bound on the first
system follows by considering the \emph{same} policy for the second
system as for the first system. 

To proceed, consider the second system, for which we denote the number of
unallocated steps of type $(j,k)$ at epoch $t$ by $\hat{Q}_{j,k}(t)$.
Now consider the set $\mc{C}$. We claim that this set is coordinate convex, i.e.,
it is a convex set and if $\mbf{a}+\mbf{\epsilon} \in \mc{C}$
for some $\mbf{\epsilon}, \mbf{a} \in \R_+^N$ then  $\mbf{a} \in \mc{C}$.
To prove this claim, we first show that the set $\mbf{C}$ is coordinate convex.

First we prove that $\mbf{C}$ is a convex subset of $\R_+^N$.
If $\mbf{a}, \mbf{a}' \in \mbf{C}$, then there exist
$\left(\mbf{a}(\bu) \in C^{\text{cvx}}(\bu): \bu \in \Z_+^{M}\right)$ and 
$\left(\mbf{a}'(\bu)\in C^{\text{cvx}}(\bu): \bu \in \Z_+^{M}\right)$ such that
\[
\sum_{\bu} \Gamma(\bu) \mbf{a}(\bu) = \vlambda\mbox{,}\quad
\sum_{\bu} \Gamma(\bu) \mbf{a}'(\bu) = \vlambda'\mbox{.}
\]

Thus for any $\gamma \in [0,1]$, 
\begin{align*}
\gamma \mbf{a} + (1-\gamma) \mbf{a}' = \sum_{\bu} \Gamma(\bu) (\gamma \mbf{a}(\bu) + (1-\gamma) \mbf{a}'(\bu).
\end{align*}
Note that $C^{\text{cvx}}(\bu)$ is convex since it is the convex hull of $C(\bu)$; hence $\gamma \mbf{a}(\bu) + (1-\gamma) \mbf{a}'(\bu) \in \C(bu)$, which in turn implies $\gamma \mbf{a} + (1-\gamma) \mbf{a}' \in \mbf{C}$. 

For coordinate convexity note that any $\mbf{a}$ is a $\Gamma(\bu)$ combination of some 
$\{\mbf{a}(\bu) \in C^{\text{cvx}}(\bu)\}$ and any $\mbf{a}(\bu)$ is some convex combination of
 elements of $C(\bu)$. Also, from the allocation constraints it is apparent 
that if $\mbf{a} \in C(\bu)$ then also $\mbf{a}' \in C(\bu)$ if $\mbf{a}' \le \mbf{a}$. 
These two imply that for any $\mbf{a}\in \bar{\C}$, if there exists
an $\mbf{a}' \le \mbf{a}$ (component-wise) and $\mbf{a}'\ge \mbf{0}$, 
then $\mbf{a}' \in \bar{\C}$. Hence, $\mbf{C}$ is coordinate convex.

Note that $\mc{C}=\{\mbf{a}: \mbf{a}^E \in \mbf{C}\}$. Note that if $\mbf{a} \le \mbf{a}'$
(coordinate-wise) then the same is true for $\mbf{a}^E$ and $\mbf{a}^{\prime E}$. Also,
if $\mbf{a}''=\gamma \mbf{a} + (1-\gamma) \mbf{a}'$ for any $\gamma \in [0,1]$, then
$\mbf{a}^{\prime\prime E}=\gamma \mbf{a}^E + (1-\gamma) \mbf{a}^{\prime E}$. This proves that
$\mc{C}$ is coordinate convex subset of $\R_+^N$.

As $\mc{C}$ is coordinate convex, for any $\vlambda$ outside the closure of $\mc{C}$ 
there exists $h \in \R_+^N$ s.t. $h^T\vlambda > \sup_{\mbf{x} \in \mc{C}} h^T \mbf{x}$.

Consider the following. Let $\mbf{\hat{Q}}=(\hat{Q}_{j,k})$, 
$\mbf{A}=(A_{j,k}(t))$ and $\mbf{\hat{D}}(t)=(D_{j,k}(t))$. Note that
as jobs in this system do not have precedence constraints,
for all $k$, $A_{j,k}=A_j(t)$.

\begin{align}
\EX\left[h^T \mbf{\hat{Q}}(t+1)\right] 
& = \EX\left[h^T \left(\mbf{\hat{Q}}(t)+\mbf{A}(t)-\mbf{\hat{D}}(t)\right)\right] \nonumber \\
& = \EX\left[h^T \left|\mbf{\hat{Q}}(t)+\mbf{A}(t)-\mbf{\Delta}(t)\right|^{+}\right] \nonumber 
\end{align}
where $\mbf{\Delta}(t)$ is the number of possible departures under the scheme if there were infinite
number of steps of each type, and $|\cdot|^{+}$ is shorthand for $\max(\cdot,0)$. 
As $|x|^{+}$ is a convex function of $x$, 
$h^T \left|\mbf{\hat{Q}}(t)+\mbf{A}(t)-\mbf{\Delta}(t)\right|^{+}$ is a convex function of
$\mbf{Q}(t), \mbf{A}(t)$, and $\mbf{\Delta}(t)$. Thus by Jensen's inequality:
\begin{align*}
&\EX\left[h^T \left|\mbf{\hat{Q}}(t)+\mbf{A}(t)-\mbf{\Delta}(t)\right|^{+}\right] \nonumber \\
&\quad \ge h^T \left|\EX\left[\mbf{\hat{Q}}(t)\right]+\EX\left[\mbf{A}(t)\right]-\EX\left[\mbf{\Delta}(t)\right]\right|^{+} \nonumber \\
& \quad \ge h^T\EX\left[\mbf{\hat{Q}}(t)\right] +  h^T \EX\left[\mbf{A}(t)\right] - 
h^T \EX\left[\mbf{\Delta}(t)\right] \mbox{.}\nonumber
\end{align*}

Note that $h^T \EX\left[\mbf{\Delta}(t)\right] \le \sup_{\mbf{x} \in \mc{C}} h^T \mbf{x}$ and
$h^T\EX\left[\mbf{A}(t)\right]=h^T \vlambda$, hence,
\begin{equation}
\EX\left[h^T \mbf{\hat{Q}}(t+1)\right] \ge h^T\EX\left[\mbf{\hat{Q}}(t)\right] + \epsilon, \ \epsilon > 0\mbox{.}
\end{equation}

So, we have $\EX[h^T \mbf{\hat{Q}}]$ to be unbounded (i.e., for any constant $B$, there
exists a $t$ s.t. $\EX[h^T \mbf{\hat{Q}}(t)]>B$) under any policy. As $h \ge 0$, we have
that $\mbf{\hat{Q}}^T \mbf{1}$ to be unbounded and hence, the system is not stable under any policy.

\subsection{Proof of Theorem \ref{thm:centralized}}
Before proceeding, we state small lemma that will be useful.
\begin{lemma}
\label{lem:xyz1}
For any $x, y, z \ge 0$, $(|x-y|^+ + z)^2 \le  x^2 + y^2 + z^2 + 2 x (z-y)$.
\end{lemma}
\begin{IEEEproof}
\begin{align*}
(|x-y|^+ + z)^2 &= (|x-y|^+)^2 + z^2 + 2 z |x-y|^+ \\ \nonumber 
&\le (|x-y|^+)^2 + z^2 + 2 x z \\ \nonumber 
&\le (x-y)^2 + z^2 + 2 x z, \\ \nonumber 
&= x^2 + y^2 + z^2 + 2 x (z-y)\mbox{,}
\end{align*}
where the last inequality follows because $(\max(0,a))^2 \le a^2$.
\end{IEEEproof}
Now for the proof of Thm.~\ref{thm:centralized}

The process $\{Q_{j,k}(t)\}$ is a discrete-time Markov chain on $\Z_+^{\sum_j K_j}$
under the centralized scheme. This is because arrival and availability processes are
i.i.d.\ and the centralized allocation at epoch $t$ does not depend on process values before 
$t$. We show that for this chain, all closed classes are positive recurrent and that 
the chain enters one of the closed classes almost surely. 
Note that this implies that starting with any initial distribution, the Markov chain 
reaches a stationary distribution (which may depend on the initial condition).  This is in the
sense that there exists a $d \in \{1, 2, \dots \}$ (as there may be a closed
class which is not aperiodic), and a distribution $\pi$ on $\Z_+^{\sum_j K_j}$
such that $\mbf{Q}(td) \to \pi$ in distribution. 

To show stability we need $\lim \sup_{t \to \infty} \EX[Q_{j,k}(t)] < 
\infty$, for all $(j,k)$. Towards this, note it is sufficient to show
$\EX_{\pi} [\sum_{j,k} Q_{j,k}]$ is finite, because this implies 
$\lim_{t\to \infty} \EX [\sum_{j,k} Q_{j,k}(td)]$ is finite. Note that for any
$1<\tau<d$:
\[
\sum_{j,k} Q_{j,k}(td+\tau) \le  \sum_{j,k} [Q_{j,k}(t)+\sum_{t'=1}^d A_j(t')] \mbox{.}
\]
Since arrivals have finite expectation, $\lim \sup_{t \to \infty} \EX[Q_{j,k}(t)] < 
\infty$, for all $(j,k)$.

Now, it is sufficient to prove that
starting with any initial distribution, there exists a $d \in \{1, 2, \cdots\}$ 
such that $\mbf{Q}(td) \to \pi$ in distribution and $\EX_{\pi} [\sum_{j,k} Q_{j,k}]$ 
is finite. To prove the convergence in distribution we use a variation 
of the Foster-Lyapunov theorem presented in \cite{TassiulasE1992}.


When $T_j$ is a directed rooted tree, we have to consider a 
Lyapunov function:
\[
L(\mbf{Q})=\sum_{j} \sum_{k} l_{j,k} Q^2_{j,k}\mbox{,}
\]
where $l_{j,k}$ is the number of leaves in the subtree of $T_j$ rooted at $k$.

Before proceeding with the proof of this case, we prove a reordering lemma for the Lyapunov function.
\begin{lemma}
\label{lem:reordering}
For any allocation $\{S_{j,k}\}$,
\begin{align}
&\sum_j\Bigg(l_{j,1} Q_{j,1} (A_{j,1}-S_{j,1}) + \sum_{k>1}l_{j,k} \left(Q_{j,k}(t)S^*_{j,p_j(k)}(t)-Q_{j,k}(t) S_{j,k}(t)\right)\Bigg) \nonumber \\
&\quad = \sum_{j} \sum_{k=1}^{K_j} 
\sum_{r \in c_j(k)} l_{j,r}(Q_{j,k}(t)-Q_{j,r}(t))(A_{j,1} - S_{j,k})\mbox{.} \label{eq:centralizedT4} 
\end{align}
\end{lemma}
\begin{IEEEproof}
First we claim that for any $j$,
\begin{align}
\sum_{k>1}l_{j,k} \left(Q_{j,k}(t)S^*_{j,p_j(k)}(t)-Q_{j,k}(t) S_{j,k}(t)\right) = \sum_{k=2}^{K_j} \sum_{r \in c_j(k)} l_{j,r}(Q_{j,k}(t)-Q_{j,r}(t)) S_{j,k}
\end{align}

This can be seen by comparing coefficients of $Q_{j,k}, k>1$ on both sides of
the expression. Note that $\sum_{r \in c_j(k)} l_{j,r}=l_{j,k}$. Also, note 
that in the sum in the right side, $Q_{j,k}$ appears twice, once in the sum
$-\sum_{r \in c_j(k)} l_{j,r}(Q_{j,k}(t)-Q_{j,r}(t)) S_{j,k}$ where the 
coefficient is $-S_{j,k}$ and again in the sum $-\sum_{r \in c_j(p_j(k))} l_{j,r}(Q_{j,p_j(k)}(t)-Q_{j,r}(t)) S_{j,p_j(k)}$
where the coefficient is $S_{j,p_j(k)}$.

This implies that for any $j$,
\begin{align}
& -l_{j,1} Q_{j,1} S_{j,1}  + \sum_{k>1}l_{j,k} \left(Q_{j,k}(t)S^*_{j,p_j(k)}(t)-Q_{j,k}(t) S_{j,k}(t)\right)\ \nonumber \\
&\quad = - \sum_{k=1}^{K_j} \sum_{r \in c_j(k)} l_{j,r}(Q_{j,k}(t)-Q_{j,r}(t)) S_{j,k} \mbox{.} \label{eq:reordering1} 
\end{align}

Note that since $l_{j,k}= \sum_{r \in c_j(k)} l_{j,r}$, 
\[
l_{j,1} Q_{j,1} = \sum_{r \in c_j(1)} l_{j,r}(Q_{j,1}-Q_{j,r}) + \sum_{r \in c_j(1)} l_{j,r} Q_{j,r}\mbox{.}
\]
Again applying the same restructuring of the terms for the subtrees rooted at $r$, we
eventually obtain:
\[
l_{j,1} Q_{j,1} = \sum_{r \in c_j(k)} l_{j,r}(Q_{j,k}-Q_{j,r})\mbox{.}
\]
The result follows by combining this (multiplied by $A_{j,1}$) with \eqref{eq:reordering1}.
\end{IEEEproof}

Now for the main proof.
\begin{align*}
\EX[L(\mbf{Q}(t+1))|\mbf{Q}(t)] = \EX\left[\sum_{j,k}(Q_{j,k}(t)-D_{j,k}(t) + A_{j,k}(t))^2|\mbf{Q}(t)\right]\mbox{.}
\end{align*}

Let $p_j(k)$ denotes the parent of $k$ in $T_j$ and $c_j(k)$ denote the set of
children of $k$ in $T_j$.
For all $t$ and $k=1$, $A_{j,1}(t)=A_j(t)$. For $k>1$, $A_{j,k}(t)=D_{j,p_j(k)}(t)$.
Note that $S^*_{j,k} \ge D_{j,k}$ for all $t$. So,
\begin{align*}
&\EX\left[\sum_{j,k}l_{j,k}(Q_{j,k}(t)-D_{j,k}(t) + A_{j,k}(t))^2|\mbf{Q}(t)\right] \nonumber \\
&= \EX\left[\sum_{j,k}l_{j,k}
(|Q_{j,k}(t)-S^*_{j,k}(t)|^+ + A_{j,k}(t))^2|\mbf{Q}(t)\right] \nonumber \\
&\le \EX\left[\sum_{j}\left(l_{j,1} (|Q_{j,1}(t)-S^*_{j,1}(t)|^+ + A_{j,1}(t))^2  + \sum_{k>1}l_{j,k}(|Q_{j,k}(t)-S^*_{j,k}(t)|^+ + S^*_{j,p_j(k)}(t))^2\right)|\mbf{Q}(t)\right]\mbox{.}
\end{align*}

By Lem.~\ref{lem:xyz1},
\begin{align}
&\EX[L(\mbf{Q}(t+1))|\mbf{Q}(t)] \nonumber \\
&=\EX\bigg[\sum_{j}l_{j,1}\left(A^2_{j,1}(t) + (S^*_{j,1}(t))^2 + Q^2_{j,1} + 2 Q_{j,1}(A_{j,1} - S^*_{j,1}) \right. \nonumber\\
&\left. \qquad + \sum_{k>1}l_{j,k}((S^*_{j,k}(t))^2 + (S^*_{j,p_j(k)}(t))^2  + Q^2_{j,k} + 2 Q_{j,k}(S^*_{j,p_j(k)} - S^*_{j,k})) \right)|\mbf{Q}(t)\bigg]  \nonumber \\
&\le C_1 + 2 \EX\left[\sum_{j,k} l_{j,1} (S^*_{j,k}(t))^2|\mbf{Q}(t)\right] \nonumber \\ 
&\quad  + 2 \EX\left[\sum_j\left(l_{j,1} Q_{j,1}(A_{j,1}-S^*_{j,1})  +\sum_{k>1} l_{j,k} Q_{j,k}(S^*_{j,p_j(k)} - S^*_{j,k})\right)|\mbf{Q}(t)\right] \label{eq:cetralizedT1} \\
&\le  C_2 + 2 \EX\left[\sum_j\left(l_{j,1} Q_{j,1}(A_{j,1}-S^*_{j,1})  +\sum_{k>1} l_{j,k} Q_{j,k}(S^*_{j,p_j(k)} - S^*_{j,k})\right)|\mbf{Q}(t)\right] \label{eq:centralizedT2}
\end{align}
Eq.~\eqref{eq:cetralizedT1} follows because arrival processes have bounded second moments and are
i.i.d.\ and the fact $(S^*_{j,k}(t))^2\ge 0$ (so, over-counting them gives an upper bound).
Eq.~\eqref{eq:centralizedT2} is due to the following (where $K=\max_j K_j$):
\begin{align}
\EX\left[\sum_{j,k}l_{j,k} (S^*_{j,k}(t))^2|\mbf{Q}(t)\right] 
&\le K \EX\left[\left(\sum_{j,k}S^*_{j,k}(t)\right)^2|\mbf{Q}(t)\right] \nonumber \\
&\le K \frac{1}{\max_{j,k} (\sum_s r_{j,k,s})^2} \EX\left[\left(\sum_{j,k,s}S^*_{j,k}(t)r_{j,k,s}\right)^2|\mbf{Q}(t)\right] \nonumber \\
&\le K \frac{1}{\max_{j,k} (\sum_s r_{j,k,s})^2} \EX\left[\left(\sum_{m,s}U_{m}(t)h_{m,s}\right)^2|\mbf{Q}(t)\right] \label{eq:centralizedT3} \\
&\le K \frac{\max_{m} (\sum_s r_{m,s})^2}{\max_{j,k} (\sum_s r_{j,k,s})^2} \EX\left[\sum_m U_m^2\right] \nonumber \\
&< \infty \mbox{.} \nonumber
\end{align}
Eq.~\eqref{eq:centralizedT3} comes from the task allocation constraint and the last step
follows as availability processes have bounded second moment.

Consider the last term of \eqref{eq:centralizedT2}, as $C_2$ plus this is the upper bound
for Lyapunov drift $\EX[L(\mbf{Q}(t+1))|\mbf{Q}(t)] - L(\mbf{Q}(t))$. Then by Lem.~\ref{lem:reordering} 
and the fact that $\{A_j(t)\}$ are i.i.d.,
\begin{align*}
& \EX\left[\sum_j\left(l_{j,1} Q_{j,1} (A_{j,1}-S^*_{j,1})  + \sum_{k}l_{j,k} \left(Q_{j,k}(t)S^*_{j,p_j(k)}(t)-Q_{j,k}(t) S^*_{j,k}(t)\right)\right)|\mbf{Q}(t)\right] \nonumber \\
&= \sum_{j} \sum_{k=1}^{K_j} \sum_{r \in c_j(k)} l_{j,r}(Q_{j,k}(t)-Q_{j,r}(t)) \lambda_j  - \EX\left[\sum_{j} \sum_{k=1}^{K_j} \sum_{r \in c_j(k)} l_{j,r}(Q_{j,k}(t)-Q_{j,r}(t)) S^*_{j,k}|\mbf{Q}(t)\right]\mbox{.} \nonumber
\end{align*}

Note that for any $\mbf{Q}(t)$ and $\mbf{U}(t)$:
\begin{align*}
\sum_{j} \sum_{k=1}^{K_j} \sum_{r \in c_j(k)} l_{j,r}(Q_{j,k}(t)-Q_{j,r}(t)) S^*_{j,k} \ge \max_{\mbf{a}\in C(\mbf{U}(t))}\sum_{j} \sum_{k=1}^{K_j} \sum_{r \in c_j(k)} l_{j,r}(Q_{j,k}(t)-Q_{j,r}(t))a_{j,k} \mbox{.}
\end{align*}

Note that in the optimal allocation $\{S^*_{j,k}\}$,  $S^*_{j,k} \ge 0$ only if
$\sum_{r \in c_j(k)} l_{j,r}(Q_{j,k}(t)-Q_{j,r}(t)))\ge 0$ (otherwise, 
just setting them to $0$ gives a better allocation). So,
\begin{align*}
\sum_{j} \sum_{k=1}^{K_j} \sum_{r \in c_j(k)} l_{j,r}(Q_{j,k}(t)-Q_{j,r}(t))S^*_{j,k} \ge \max_{\mbf{a}\in C(\mbf{U}(t))}\sum_{j} \sum_{k=1}^{K_j} |\sum_{r \in c_j(k)} l_{j,r} (Q_{j,k}-Q_{j,r})|^+ a_{j,k}\mbox{.}
\end{align*}

Hence, 
\begin{align*}
& \EX\left[\sum_{j} \sum_{k=1}^{K_j} \sum_{r \in c_j(k)} l_{j,r} (Q_{j,k}(t)-Q_{j,r}(t))S^*_{j,k}|\mbf{Q}(t)\right] \nonumber \\
&= \EX\left[\sum_{j} \sum_{k=1}^{K_j} |\sum_{r \in c_j(k)} l_{j,r} Q_{j,k}-Q_{j,r}|^+ S^*_{j,k}|\mbf{Q}(t)\right] \nonumber \\ 
&\ge \sup_{\mbf{a} \in \mbf{C}} \EX\left[\sum_{j} \sum_{k=1}^{K_j} |\sum_{r \in c_j(k)} l_{j,r} (Q_{j,k}-Q_{j,r})|^+ a_{j,k}|\mbf{Q}(t)\right] \nonumber \\
&\ge\sup_{\begin{subarray}{c}\mbf{a} \in \mbf{C}: \\
 a_{j,k}=a_j 1 \le k \le K_j\end{subarray}} 
\EX\left[\sum_{j} \sum_{k=1}^{K_j} 
|\sum_{r \in c_j(k)} l_{j,r}
Q_{j,k}-Q_{j,r}|^+ a_{j,k}|\mbf{Q}(t)\right] \nonumber \\
&= \sup_{\mbf{a} \in \mc{C}}\EX\left[\sum_{j} \sum_{k=1}^{K_j} 
|\sum_{r \in c_j(k)} l_{j,r}Q_{j,k}-Q_{j,r}|^+ a_{j}|\mbf{Q}(t)\right] \nonumber \\
&\ge \EX\left[\sum_{j} \sum_{k=1}^{K_j} 
|\sum_{r \in c_j(k)} l_{j,r}Q_{j,k}-Q_{j,r}|^+ \lambda_{j}|\mbf{Q}(t)\right]  + \epsilon \sum_{j}\sum_{k=1}^{K_j} |\sum_{r \in c_j(k)} l_{j,r} (Q_{j,k}-Q_{j,r})|^+, \nonumber
\end{align*}
because $\vlambda+\epsilon \in \mc{C}$.

As, 
\begin{align*}
&\sum_{j} \sum_{k=1}^{K_j} |\sum_{r \in c_j(k)} l_{j,r} Q_{j,k}-Q_{j,r}|^+\lambda_{j} \ge \sum_{j} \sum_{k=1}^{K_j} \sum_{r \in c_j(k)} l_{j,r}(Q_{j,k}(t)-Q_{j,r}(t))\lambda_{j}
\end{align*}
we have
\begin{align*}
\EX[L(\mbf{Q}(t+1)) - L(\mbf{Q}(t)) | \mbf{Q}(t)] \le C_2 - \epsilon \sum_{j} \sum_{k=1}^{K_j} |\sum_{r \in c_j(k)} l_{j,r}(Q_{j,k}-Q_{j,r})|^+\mbox{.} \nonumber
\end{align*}

Note that for $\{Q_{j,k}\}$ sufficiently large, 
$\sum_{j} \sum_{k=1}^{K_j} |\sum_{r \in c_j(k)} l_{j,r}(Q_{j,k}(t)-Q_{j,r}(t))|^+$ is also large.
This is because if $\{Q_{j,k}\}$ is larger than $B$ (in max-norm) then 
there exists a $j$ such that $\max_k Q_{j,k}>B$. Now consider
two cases, if $Q_{j,L_j}\ge\frac{B}{2}$ for some leaf node $L_j$
then we have the drift $\le C_2 - \epsilon \frac{B}{2}$
which can be made strictly negative by choosing $B$ appropriately.

If $Q_{j,L_j}<\frac{B}{2}$ for all leaf nodes, then there exists a $k_0$ such that $Q_{j,k_0}>B$.
Note the following for the set of nodes $T_{k_0}$ in the subtree rooted at $k_0$ and
$L_{k_0}$ being the leaves of $T_{k_0}$:
\begin{equation}
\sum_{k\in T_{k_0}} \sum_{r \in c_j(k)} l_{j,r}(Q_{j,k}-Q_{j,r}) = \sum_{l \in L_{k_0}} (Q_{j,k_0}-Q_{j,l})\mbox{.} \label{eq:centralizedT5}
\end{equation}

Hence, we have that $\sum_{k\in T_{k_0}} 
\sum_{r \in c_j(k)} l_{j,r}(Q_{j,k}-Q_{j,r}) \ge l_{j,k_0} \frac{B}{2} \ge \frac{B}{2}$.

Thus we show strictly negative drift for sufficiently large $\{Q_{j,k}\}$ and
the drift is bounded by $C_2 < \infty$. Hence, by the Foster-Lyapunov theorem
in \cite{TassiulasE1992} we have that for any initial distribution, there
exists a $d \in \{1, 2, \ldots\}$ 
such that $\mbf{Q}(td) \to \pi$ in distribution. 

To prove finite expectation we consider
the following.
\begin{align*}
&\EX[L(\mbf{Q}(t+1)) - L(\mbf{Q}(t)) | \mbf{Q}(t)] \le C_2 - \epsilon \sum_{j} \sum_{k=1}^{K_j} |\sum_{r \in c_j(k)} l_{j,r} (Q_{j,k}(t)-Q_{j,r}(t))|^+
\end{align*}
which implies that
\begin{align*}
\EX[L(\mbf{Q}(t+1)) - L(\mbf{Q}(t))] \le C_2 - \epsilon \EX\left[\sum_{j} \sum_{k=1}^{K_j} |\sum_{r \in c_j(k)} l_{j,r}(Q_{j,k}(t)-Q_{j,r}(t))|^+\right]\mbox{.} 
\end{align*}

Summing both sides from $0$ to $T$, we get:
\begin{align*}
\frac{1}{T} \sum_{t=1}^T \EX\left[\sum_{j} \sum_{k=1}^{K_j} |\sum_{r \in c_j(k)} l_{j,r} Q_{j,k}(t)-Q_{j,r}(t)|^+\right]\le \frac{1}{\epsilon} \left(C_2 - \frac{1}{T} \EX[L(\mbf{Q}(T+1))] + \EX[L(\mbf{Q}(0))]\right)\mbox{.} 
\end{align*}

As $\EX[L(\mbf{Q}(0))]$ finite, for any initial condition we have
\[
\frac{1}{T} \sum_{t=1}^T \EX\left[\sum_{j} \sum_{k=1}^{K_j} |\sum_{r \in c_j(k)} l_{j,r}(Q_{j,k}(t)-Q_{j,r}(t))|^+\right]<C_3\mbox{,}
\]
for all $T$.

As all terms are positive, for any $d\in\{1, 2, \dots\}$, 
\begin{align*}
\lim_{T\to \infty} \frac{d}{T} \sum_{t=1}^T \EX\left[\sum_{j} \sum_{k=1}^{K_j} 
|\sum_{r \in c_j(k)} l_{j,r}(Q_{j,k}(td)-Q_{j,r}(td))|^+
\right] < C_3\mbox{.}
\end{align*}

By the ergodicity of a Markov chain in a positive recurrent class this implies that 
\[
\EX_{\pi}\left[\sum_{j} \sum_{k=1}^{K_j} 
|\sum_{r \in c_j(k)} l_{j,r}(Q_{j,k}(td)-Q_{j,r}(td))|^+\right] < C_3\mbox{.}
\]

This proves that $\EX_{\pi}\left[Q_{j,L_j}\right]<C_3$ for any leaf node $l_j$.

By \eqref{eq:centralizedT5} we have that for any $k \in T_j$,
\begin{align*}
l_{j,k} Q_{j,k} 
&\quad= \sum_{l_j} Q_{j,l_j} + \sum_{k' \in T_k} \sum_{r \in c(k')} l_{j,r} (Q_{j,k'} - Q_{j,r}) \nonumber \\
&\quad\le\sum_{l_j} Q_{j,l_j} + \sum_{k' \in T_k} |\sum_{r \in c(k')} l_{j,r} (Q_{j,k'} - Q_{j,r})|^+ \mbox{.} 
\end{align*}
Hence, it follows that $\EX_{\pi}\left[Q_{j,k}\right] < \infty$.
This implies that $\EX_{\pi}\left[\sum_{j} \sum_{k=1}^{K_j} Q_{j,k}\right] < \infty$
and so the proof is complete.

\subsection{Proof of Theorem \ref{thm:LPrelax}}
In deriving \eqref{eq:centralizedT2} we did not use any property of the allocation
$\{S^*_{j,k}\}$ other than the fact that it has to satisfy the step allocation constraint.
Hence, this upper bound for Lyapunov drift is valid for any arbitrary feasible allocation
$\{S_{j,k}\}$.

Hence, under the LP-relaxation base allocation $\{{S}^R_{j,k}\}$ by Lem.~\ref{lem:reordering} we have:
\begin{align}
\EX[L(\mbf{Q}(t+1)) - L(\mbf{Q}(t)) | \mbf{Q}(t)] \le C_2 + 2 \EX\left[\sum_j\sum_{r\in c_j(k)} l_{j,r} \left(Q_{j,k}-Q_{j,r}\right) {S}^R_{j,k}|\mbf{Q}(t)\right]\mbox{.} 
\label{eq:centralizedRelax1}
\end{align}

Note that for the optimum of the problem in \eqref{eq:centAlgoLP}, $\{{S}^R_{j,k}(t)\}$,
the following is true.
\begin{align*}
&\sum_{j} \sum_{k=1}^{K_j} \sum_{r \in c_j(k)} l_{j,r}(Q_{j,k}(t)-Q_{j,r}(t)) S^*_{j,k} \nonumber \\
&\quad= \sum_{j} \sum_{k=1}^{K_j} |\sum_{r \in c_j(k)} l_{j,r}(Q_{j,k}(t)-Q_{j,r}(t))|^+ S^*_{j,k} \nonumber \\
&\quad\le \sum_{j} \sum_{k=1}^{K_j} |\sum_{r \in c_j(k)} l_{j,r}(Q_{j,k}(t)-Q_{j,r}(t))|^+ \hat{S}_{j,k}
\end{align*}
This is because \eqref{eq:centAlgoLP} solves a relaxed problem and the
optimal allocation has $S^*_{j,k}=\hat{S}_{j,k}=0$ for 
$\sum_{r \in c_j(k)} l_{j,r}(Q_{j,k}(t)-Q_{j,r}(t))$. As
$S^R_{j,k} = \lfloor \hat{S}_{j,k} \rfloor$ we have that
\begin{align*}
\sum_{j} \sum_{k=1}^{K_j} |\sum_{r \in c_j(k)} l_{j,r}(Q_{j,k}(t)-Q_{j,r}(t))|^+ S^*_{j,k} \le \sum_{j} \sum_{k=1}^{K_j} |\sum_{r \in c_j(k)} l_{j,r}(Q_{j,k}(t)-Q_{j,r}(t))|^+ ({S}^R_{j,k}+1)
\end{align*}

Hence, for any $\vlambda$ such that $\vlambda+ \mbf{1}(1+\epsilon) \in \mc{C}$ using the
same proof as above we can show that the system is stable.

\subsection{Proof of Theorem \ref{thm:prioGreedy}}

This proof has the following structure. As the total number of incomplete 
jobs is equal to the total number of unallocated steps, we first show that the 
total number of unallocated steps at depth $0$ (i.e., at the root of
each $T_j$) across all types have the desirable property. Then 
we show that this property propagates.

\noindent \emph{Proof for Depth-$0$ Steps:}

This part is same as the proof \cite{ChatterjeeVV2015_arXiv} of performance guarantee of GreedyJob algorithm in \cite{ChatterjeeVV2015}. We present it here for the sake of completeness.

Consider the different types of unallocated steps at depth $0$.
These are given by $\{Q_{j,1}(t): j \in [N]\}$. 

Consider the following processes: for each $s \in [S]$,
$Q^s_1(t)=\sum_{j:r_{j,1,s}>0} Q_{j,1} r_{j,1,s}$ which represent
the number of unserved hours of skills $s$ for all steps at
depth $0$. 

We now construct another process $\tilde{Q}_1$ such that it dominates
the process $\sum_s Q^s_1$. So, if we can upper bound $\tilde{Q}_1$,
then the same bound applies for $\sum_s Q^s_1$. Hence, in turn we get a
bound for $\{Q_{j,1}(t)\}$ (since $\min\{r_{j,k,s}>0\}=\Theta(1)$ by the
assumption that $\{r_{j,k,s}\}$ do not scale with the system size).

To construct a suitable $\tilde{Q}_1$, we make the following observation
about the dynamics of $Q^s_1$ and $\{Q_{j,1}\}$. 
At each time $t$,
$\sum_j A_{j,1}(t) r_{j,1,s}$ amount of $s$ skill-hour is brought to add
to $Q^s_1$. Also, this queue gets some service depending on the available
agent hours.

At time $t$, $\sum_m U_m(t) h_{m,s}$ $s$-skill hour of service is brought
by the agents. 

For a step to be allocated, all of its tasks must find an allocation.
Hence, for a step in type $j$-job to find an allocation it must get $r_{j,1,s}$
hours of service from each skill $s$. Thus at any time $t$ any skill $s$ queue
gets a service of at least
\[
\min_{s \in [S]} \sum_m U_m h_{m,s} - \bar{r}\mbox{,}
\]
where $\bar{r}=\max\{r_{j,k,s}\}$, due to the following. 
For each skill, $\sum_s U_m h_{m,s}$ hours are available. Note that a
step can be allocated if all its tasks find allocations, the converse of which is also
true. That is, if all tasks of a step find allocation, then the step can be 
allocated. As $\min_{s \in [S]} \sum_s U_m h_{m,s}$
hours of service are brought by the agents for each skill, at 
least $\min_{s \in [S]} \sum_s U_m h_{m,s} - \bar{r}$ of $s$-skill
hours are served (because a maximum of $\bar{r}$ can be wasted, as no
task is of size more than $\bar{r}$). Note that as depth $d$ steps
have priority in Priority Greedy algorithm over steps at depth $\ge d+1$, 
they do not have to share resource with higher-depth steps. 
So at depth $d$, $\min_{s \in [S]} \sum_s U_m h_{m,s}$
is available for service to steps at depth $\le d$.

Also, note that the amount of required service brought to the 
queue $Q^s_1$ at time $t$ is upper-bounded by
\[
\max_{s \in [S]} \sum_j A_{j,1}(t) r_{j,1,s}\mbox{.}
\]

Consider a process $\tilde{Q}^s_1$ with evolution:
\begin{align*}
\tilde{Q}^s_1(t+1) & = \max(\tilde{Q}^s_1(t) +  \max_{s \in [S]} \sum_j A_{j,1}(t) r_{j,1,s}  - \min_{s \in [S]} \sum_m U_m h_{m,s} + \bar{r}, 0)\mbox{.} 
\end{align*}


Note that given $\tilde{Q}^s_1(t_0)\ge Q^s_1(t_0)$ at some $t_0$, the same holds true
for all $t \ge t_0$. This is because for $x, a, b \ge 0$ and $x', a', b' \ge 0$, with
$x\ge x'$, $a\ge a'$ and $b\le b'$
\[
\max(x+a-b,0) \ge \max(x'+a'-b',0)\mbox{,}
\]
and so the monotonicity propagates over time.

To bound $\sum_s Q^s_1$, it is sufficient to bound $\sum_s \tilde{Q}^s_1(t)$.
Note that each $\tilde{Q}^s_1$ has exactly the same evolution, so let us
consider 
\[
\tilde{Q}_1:=S\tilde{Q}^1_1\mbox{,}
\] 
which bounds $\sum_s Q^s_1$.

From the evolution,
\begin{align*}
\tilde{Q}_1(t+1) &= \max(\tilde{Q}_1(t) +  S \max_{s \in [S]} \sum_j A_{j,1}(t) r_{j,1,s} - S \min_{s \in [S]} \sum_m U_m h_{m,s} + \bar{r}, 0)\mbox{,} 
\end{align*}
and we can write the Loynes' construction for this process which has the same distribution as
the following process (and for simplicity we use the same notation, as we are interested in the distribution):
\begin{align*}
\tilde{Q}^1_1(0) &= \max_{\tau\le 0}
\sum_{\tau\le t \le 0}(S \max_{s \in [S]} \sum_j A_{j,1}(t) r_{j,1,s} - S \min_{s \in [S]} \sum_m U_m h_{m,s} + \bar{r})\mbox{,}
\end{align*}
assuming that the process started at $-\infty$.

Let us define $X_s(t)$ and $Y_s(t)$ as follows:
$X_s(t):=\sum_j A_{j,1}(t) r_{j,1,s}$ and 
$Y_s(t):=\sum_m U_m h_{m,s}$.
Then, 
\[
\tilde{Q}^1_1(0) = \max_{\tau\le 0} \sum_{\tau\le t \le 0}S(\max_{s} X_s(t) - \min_s Y_s(t) + \bar{r})\mbox{.}
\]

Now, for any $\theta>0$:
\begin{align}
\PR(\sum_{j} Q_{j,1} > \bar{r} q) &\le \PR(\sum_s Q^s_1 > q) \nonumber \\
&\le \PR(\tilde{Q}_1(0) > q)  \nonumber \\
&= \PR( \theta \tilde{Q}_1(0) > \theta q) \label{eq:prioGreedy1} \\
&= \PR(\exp(\theta \tilde{Q}_1(0)) > \exp(\theta q)) \nonumber \\
&\le \EX[\exp(-\theta q)] \EX[\exp(\theta \tilde{Q}_1(0))]\mbox{.} \nonumber
\end{align}
Now,
\begin{align}
\EX[\exp(\theta \tilde{Q}_1(0))] &= \EX\left[\exp(\theta S \left(\max_{\tau\le 0}
\sum_{\tau\le t \le 0}(\max_{s} X_s(t) - \min_s Y_s(t) + \bar{r})\right))\right] \nonumber \\
&\le\sum_{\tau\le 0} \EX\left[\exp(\theta S 
\sum_{\tau\le t \le 0}(\max_{s}X_s(t) - \min_s Y_s(t) + \bar{r}))\right],\label{eq:prioGreedy2}
\end{align}
where inequality \eqref{eq:prioGreedy2} 
follows because for any random variables $\{Z_j\}$,
$\exp(\theta \Z_j)$ are positive random variables and the sum of positive values is
more than their maximum.

Next, we bound the term within the summation over $\tau\le 0$ in \eqref{eq:prioGreedy2}.
\begin{align}
\EX[\exp(\theta S
\sum_{\tau\le t \le 0}(\max_{s} X_s(t) - \min_s Y_s(t) + \bar{r}))] \le \prod_{\tau \le t \le 0} \EX[\exp(\theta (\max_{s}X_s(t) - \min_s Y_s(t) + \bar{r})))] \mbox{,}
\label{eq:prioGreedy3}
\end{align}
which follows because $X_s(t)$, $Y_s(t)$ are i.i.d.\ over time.

Next we bound the term within the product in \eqref{eq:prioGreedy3},
\begin{align}
\EX\left[e^{\theta S  \left(\max_s X_s(t) - \min_s Y_s(t) + \bar{r}\right)}\right] \le \sum_{s,s'} \EX \left[e^{\theta S \left(X_s(t) - Y_{s'}(t) + \bar{r}\right)}\right] \mbox{,} \label{eq:prioGreedy4x} 
\end{align}
where this follows for the same reason as \eqref{eq:prioGreedy2}.

The following lemma regarding an outer bound to the capacity region will be useful later.
\begin{lemma}
\label{lem:outerIF}
Let $\mc{C}_{I,F}^O=\{\vlambda: \sum_{j,k} \lambda_j  r_{j,k,s} \le \sum_m u_m h_{m,s}\}$. Then $\mc{C}_{I,F} \subset \mc{C}_{I,F}^O$
\end{lemma}
\begin{proof}
Consider any $\vlambda \in \mc{C}_{I,F}$. Then by definition of $\mc{C}_{I,F}$, 
$\vlambda^E \in \mbf{C}$ and
there exists $\mbf{c}(\bu) \in conv {C}_{I,F}(\bu)$ such that
\[
\vlambda^E \le \sum_{\bu} \Gamma(\bu) \mbf{c}(\bu)\mbox{.}
\]

Also, note that for each $\mbf{c}(\bu) \in conv {C}_{I,F}(\bu)$, there exists an
$\{\alpha_v(\bu)\ge 0, 1 \le v \le V_{\bu}: \sum_{v=1}^{V_{\bu}} \alpha_k=1\}$ and
$a_v(\bu) \in C_{I,F}(\bu), 1 \le v \le V_{\bu}$ such that
\[
\sum_{v=1}^{V_{\bu}}a_k(\bu) \alpha_v(\bu) = \mbf{c}(\bu)\mbox{.}
\]

Note that as $a_v(\bu) \in C_{I,F}(\bu)$ so by allocation constraint
\[
\sum_{j,k} a_{v,j,k}(\bu) r_{j,k,s} \le \sum_m u_m h_{m,s}\mbox{,}
\]
where $a_{v,j,k}$ are the number of $(j,k)$ steps that have been allocated under
$a_v(\bu)$ allocation. This in turn implies that:
\begin{align*}
& \sum_{j,k} a_{v,j,k}(\bu) \alpha_v(\bu) r_{j,k,s} \le \sum_m u_m h_{m,s} \nonumber \\
&\implies \sum_{j,k} \mbf{c}_{j,k}  r_{j,k,s} \le \sum_m u_m h_{m,s} \nonumber \\
&\implies \sum_{j,k} \lambda_j  r_{j,k,s} \le \sum_m u_m h_{m,s}\mbox{.}
\end{align*}

Hence, $\vlambda \in \mc{C}_{I,F}$ implies that $\vlambda \in \mc{C}^O_{I,F}$. 
\end{proof}


Let $\mc{A}_j(\theta)=\EX\left[e^{\theta A_j(t)}\right]$
and $\mc{U}_m(\theta)=\EX\left[e^{\theta U_m(t)}\right]$ for $j \in [N]$ and $m \in [M]$.
For $\theta \in \R$, then,
\begin{align}
\EX\left[e^{\theta (X_s(t)-Y_{s'}(t))}\right]  
&\quad= \EX\left[e^{\theta X_s(t)}\right] \EX\left[e^{-\theta Y_{s'}(t)}\right] \nonumber \\
&\quad=\EX\left[e^{\theta \sum_j A_j(t) r_{j,1,s} + \bar{r}}\right]   
\EX\left[e^{-\theta \sum_i U_i(t) h_{i,s}}\right]   \nonumber \\
&\quad= e^{\theta \bar{r}} \prod_j\EX\left[e^{\theta A_j(t) r_{j,1,s}}\right]  \prod_i\EX\left[e^{-\theta U_i(t) h_{i,s'}}\right] \label{eq:GA1} \\
&\quad= e^{\theta \bar{r}} \prod_j \mc{A}_j(\theta r_{j,1,s}) \prod_i \mc{U}_i(-\theta h_{i,s'}) \nonumber \\
&\quad= \exp \left(\theta \bar{r} + \sum_j \log \mc{A}_j(\theta r_{j,1,s}) + \sum_i \log \mc{U}_i(-\theta h_{i,s'}) \right)\mbox{.} \nonumber
\end{align}

Note that as $\vlambda \in \alpha \mc{C}$, by the definition of $\mc{C}^O_{I,F}$,
$\sum_j \lambda_j r_{j,1,s} <\alpha \sum_m \mu_m h_{m,s}$ and by assumption
$|\sum_m \mu_m h_{m,s} - \sum_m \mu_m h_{m,s'}|\le \mbox{subpoly}(N)$ which is used
in the following.

First consider the Gaussian-dominated case. Since the process variance is no more than the mean and
the moment generating function of the variance is upper-bounded by that of a zero-mean Gaussian:
\begin{align*}
\log \mc{A}_j(\theta r_{j,s}) &\le \lambda_j \theta r_{j,1,s} + \lambda_j \frac{\left(\theta r_{j,1,s}\right)^2}{2} \\
\log \mc{U}_i(-\theta h_{i,s}) &\le - \mu_j \theta h_{i,s} + \mu_j \frac{\left(\theta h_{i,s}\right)^2}{2}\mbox{.} 
\end{align*}

Note that for any two functions $k_1 x^2$ and $k_2 x$, $\lim_{x \to 0} k_2 x/k_1 x^2 = \infty$, and
hence for any $\epsilon\in(0,1)$ there exists $x^*>0$ such that for all $x<x^*$, $k_1 x^2/k_2 x <
\epsilon$. Hence for any $\epsilon \in (0,1)$,
there exist $\theta^*_{j,s}, \theta^*_{i,s}>0$, for all $i, j, s$ such that for all 
$\theta < \theta^*=\min_{i,j,s} (\theta^*_{j,s},\theta^*_{i,s})$, 
\begin{align}
\log \mc{A}_j(\theta r_{j,1,s}) &\le \lambda_j \theta^* r_{j,1,s} (1+\epsilon) \label{eq:GA2}\\
\log \mc{U}_i(-\theta h_{i,s}) &\le - \mu_i \theta h_{i,s} (1-\epsilon)\mbox{.} \label{eq:GA2A}
\end{align}

Note that since $N$, $S$, and $M$ are finite and $\theta^*_{j,s}, \theta^*_{i,s}>0$, 
for all $i, j, s, \theta^*>0$. Moreover, note that $\theta^*$ does not depend on $\vlambda, \bm{\mu}$ since the ratio of 
the linear and quadratic terms in the log moment generating functions are independent 
of $\vlambda$ and $\bm{\mu}$.

As $e^{\theta} - 1 = \sum_{k=1}^\infty \frac{\theta^k}{k!}$, for the Poisson-dominated case
we have:
\begin{align*}
\log \mc{A}_j(\theta r_{j,s}) &\le \lambda_j \sum_k \frac{\left(\theta r_{j,1,s}\right)^k}{k!} \\
\log \mc{U}_i(-\theta h_{i,s}) &\le \mu_j \sum_k \frac{\left(-\theta h_{i,s}\right)^k}{k!} \mbox{.}
\end{align*}

Again, by the same argument, we can have a $\theta^*$ for which (\ref{eq:GA2}) and (\ref{eq:GA2A})
are satisfied. As $|\sum_i \mu_i h_{i,s} - \sum_i \mu_i h_{i,s'}|=o(N^{\delta}),\forall \delta>0$,
and $\sum_i \mu_i h_{i,s}=\Omega(N^c), c>0$, for all $\theta<\theta^*$ we have:
\begin{align}
\EX\left[e^{\theta (X_s(t)-Y_{s'}(t))}\right] 
&\le \exp \left(\theta^* \bar{r} + \sum_j \lambda_j \theta^* r_{j,1,s} (1+\epsilon)  - \sum_i \mu_i \theta h_{i,s} (1-\epsilon) + \theta^* o(\mu_i \theta h_{i,s})\right) \nonumber \\
& \le \exp \left(\theta \left(\bar{r} - \sum_i \mu_i h_{i,s}(\alpha - 2 \epsilon)\right)\right) \mbox{.}
\label{eq:GA3} 
\end{align}
Note \eqref{eq:GA3} follows from the fact $\lambda \in (1-\alpha) \C^{O}$. As $\epsilon>0$ can be chosen
arbitrarily small, we can have $\alpha-2\epsilon>0$. Since $\sum_i \mu_i h_{i,s}>\sum_j (1-\alpha)\lambda_j r_{j,s} $ and $\sum_j \lambda_j r_{j,1,s}$ scales with $\lambda(N)$, for sufficiently large $\lambda_{\alpha}$
with $\lambda_j \ge \lambda_{\alpha}$ for all $j$, we have 
$\bar{r} - \sum_i \mu_i h_{i,s}(\alpha - \epsilon) \le - \gamma \sum_i \mu_i h_{i,s}(\alpha - \epsilon)$,
for some $\gamma>0$. Thus, we have for some $\theta>0$,
\begin{equation}
\EX\left[e^{\theta S (X_s(t)-Y_{s'}(t))}\right] \le \exp \left(-\theta S K(N) \right) \mbox{,}
\end{equation}
where $K(N)$ scales with $N$ no slower than 
$\sum_{s:r_{j,1,s}>0} \lambda_j(N) =\Omega(N^c)$, $c>0$.

Thus,
\[
\EX\left[e^{\theta S (\max_s X_s(t) - \min_s Y_s(t))}\right] \le S^2 \exp \left(-\theta S K(N) \right) \mbox{.}
\]
Hence, from \eqref{eq:prioGreedy2}, \eqref{eq:prioGreedy3}, and \eqref{eq:prioGreedy4x}
we have that
\begin{align*}
\EX[\exp(\theta^* \sum_s \tilde{Q}^s_1(0))] &= \EX[\exp(\theta^* S \tilde{Q}^1_1(0))] \nonumber \\
&= \EX[\exp(\theta^* S \tilde{Q}_1(0))] \nonumber \\
&\le \sum_{\tau \le 0} S^{2|\tau|}\exp(-\theta^* S K(N) |\tau|) \nonumber \\
&\le c'\mbox{,} \nonumber
\end{align*}
because $S^2 < \exp(\theta^* S K(N))$ for sufficiently large $N$.

Note that though we proved $\EX[\exp(\theta {Q}_1(t))]<c'$ for $t=0$, this holds for
any finite $t$ (exactly the same proof). This in turn implies that the number of unallocated
steps in depth $0$ have bounded exponential moment for some $\theta>0$. This will be used
in the remainder of the proof where we show that the same is true for all depths.


\noindent {\em Induction over Depths, $d$ to $d+1$:}

Now we show that if the total number of unallocated steps at depth $d$ satisfies
$\EX[\exp(\theta {Q}(0))]<c'$, then the same is true for $d+1$.
To show the same result for steps at all depths we consider the following process.
Let $d_j(k)$ be the depth of $k$ in $T_j$, then
$Q^s_{d+1}(t)=\sum_{j,k:d_j(k) \le {d+1} } Q_{j,k} r_{j,k,s}$ represents
the number of unserved hours of skills $s$ for all steps in
the system.

Like in the case of the proof for depth $0$, we 
construct  process $\tilde{Q}^s_{d+1}$ such that $\sum_s \tilde{Q}^s_{d+1}$
dominates the process $\sum_s Q^s_{d+1}$. 
Using the same argument as before, at any time 
$t$ any skill $s$ queue gets a service of at least
\[
\min_{s \in [S]} \sum_m U_m h_{m,s} - \bar{r}\mbox{,}
\]
and the amount of required service brought to the 
queue $Q^s$ at time $t$ is upper-bounded by
\[
\max_{s \in [S]} \sum_{j,k:d_j(k) \le {d+1}} A_{j,k}(t) r_{j,k,s}\mbox{.}
\]

Then using the same argument, the process
\begin{align*}
\tilde{Q}_{d+1}(t+1) 
= \left|(\tilde{Q}_{d+1}(t) +  S \max_{s \in [S]} \sum_{j,k:d_j(k) \le {d+1}}  A_{j,k}(t) r_{j,k,s} -  \min_{s \in [S]} \sum_m U_m h_{m,s} + \bar{r}\right|^+ 
\end{align*}
upper-bounds the process $\sum_s Q^s_{d+1}$. Then we can follow the steps that we followed
using $X_s$ and $Y_s$ previously.
Let $X'_s:=\sum_{j,k:d_j(k) \le {d+1}} A_{j,k}(t) r_{j,k,s}$ and $Y'_s:=\sum_m U_m h_{m,s}$, respectively.
But note that $A_{j,k}$ for $k>1$ is not an external i.i.d.\ process, rather it is the
number of steps of type $A_{j,p_j(k)}$ that were completed. Hence, we cannot 
follows the exactly same steps. Note that
\begin{align*}
\EX[\exp(\theta S \sum_{\tau\le t \le 0}(\max_{s} X'_s(t) - \min_s Y'_s(t) + \bar{r}))] 
&\le\EX[\exp(\theta \sum_{\tau\le t \le 0}\max_{s,s' \in [S]}(X'_s(t) - Y'_{s'}(t) + \bar{r}))] \nonumber \\
&\le \EX[\sum_{s,s' \in [S]} \exp(\theta \sum_{\tau\le t \le 0}(X'_s(t) - Y'_{s'}(t) + \bar{r}))] \nonumber \\
&= \sum_{s,s' \in [S]} \EX[\exp(\theta S \sum_{\tau\le t \le 0}(X'_s(t) - Y'_{s'}(t) + \bar{r}))] \mbox{.}
\end{align*}
Also note that,
\begin{align*}
\sum_{\tau\le 0} \EX[\exp(\theta S \sum_{\tau\le t \le 0}\max_{s,s' \in [S]} (X'_s(t) - Y'_{s'}(t) + \bar{r}))] \le \sum_{s,s' \in [S]} \sum_{\tau\le t \le 0} \EX[\exp(\theta S \sum_{\tau\le t \le 0}(X'_s(t) - Y'_{s'}(t) + \bar{r}))]\mbox{.}
\end{align*}

So, we investigate $\EX[\exp(\theta \sum_{\tau\le t \le 0}(X'_s(t) - Y'_{s'}(t) + \bar{r}))]$.
\begin{align*}
&\EX[\exp(\theta \sum_{\tau\le t \le 0}(X'_s(t) - Y'_{s'}(t) + \bar{r}))] \nonumber \\
&= \exp(\bar{r}\theta)\EX[\exp(\theta \sum_{\tau\le t \le 0}(\sum_{j,k} A_{j,k}(t) r_{j,k,s} - \sum_m U_m(t) h_{m,s'}))] \nonumber \\
&= \exp(\bar{r}\theta) \EX[\exp(\theta (\sum_{j,k: d_j(k)\le d+1} r_{j,k,s} \sum_{\tau\le t \le 0}A_{j,k}(t)  - \sum_m \sum_{\tau\le t \le 0} U_m(t) h_{m,s'}))]\mbox{.}
\end{align*}
Note that $\sum_{\tau\le t \le 0}A_{j,k}(t)$ represent the creation (or appearance/arrival) of 
steps of type $(j,k)$ between time $\tau$ and $0$ (with a similar interpretation for
agents in case of $\sum_{\tau\le t \le 0} U_m(t)$), which we denote by
$A_{j,k}(\tau:0)$ (and $U_m(\tau:0)$), respectively.

Now there is an important observation about $A_{j,k}(\tau:0)$:
\begin{equation}
A_{j,k}(\tau:0) \le Q_{j,p_j(k)}(\tau-1) + A_{j,p_j(k)}(\tau-1:-1)\mbox{,} \label{eq:prioGreedy5} 
\end{equation}
where $p_j(k)$ is the parent of $k$ in $T_j$, due to the following. As each job takes one slot to 
be served, no job whose step $(j,p_j(k))$ completed after 
$-1$ can have its step $(j,k)$ be available for service at or before $0$.
Thus by induction on the function $p_j$ we can write
\begin{align}
A_{j,k}(\tau:0) &\le \sum_{w=1}^{d} Q_{j,w}(\tau-1-d+w) + A_{j,1}(\tau-d_j(k):-d_j(k))\mbox{,} \label{eq:prioGreedy5}
\end{align}
as $d_j(k)=d+1$ by the inductive assumption. Note that from $1$ to $k$ (at depth $d+1$) there
is a unique $d$-length path and hence, on that path w.l.o.g.\ we denote the respective steps
by $(j,w)$ where $w$ is its depth on that path.
 
Hence,
\begin{align*}
&\EX[\exp(\theta (\sum_{j,k:d_j(k)\le d+1} r_{j,k,s} \sum_{\tau\le t \le 0}A_{j,k}(t)  - \sum_m \sum_{\tau\le t \le 0} U_m(t) h_{m,s'}))] \nonumber \\
&\le \EX[\exp(\theta (\sum_{j,k:d_j(k)\le d+1} r_{j,k,s} (\sum_{w=1}^{d} Q_{j,w}(\tau-1-d+w) \nonumber \\
& \quad + A_{j,1}(\tau-d_j(k):-d_j(k))) - \sum_m \sum_{\tau\le t \le 0} U_m(t) h_{m,s'}))]
\end{align*}

Note that $\sum_m \sum_{\tau\le t \le 0} U_m(t) h_{m,s'}$ is independent of 
\begin{align*} 
& \sum_{j,k:d_j(k)\le d+1} r_{j,k,s} \Bigg(\sum_{w=1}^{d} Q_{j,w}(\tau-1-d+w) + A_{j,1}(\tau-d_j(k):-d_j(k))\Bigg)\mbox{,}
\end{align*}
because $A_{j,1}$ are i.i.d.\ (independent of $U_m$) 
and $Q_{j,w}(\tau-d+w)$ does not depend on $U_m(\tau:0)$ for $d\ge w \ge 1$.
Hence, 
\begin{align*}
&\EX[\exp(\theta (\sum_{j,k:d_j(k)\le d+1} r_{j,k,s} \sum_{\tau\le t \le 0}A_{j,k}(t)  - \sum_m \sum_{\tau\le t \le 0} U_m(t) h_{m,s'}))] \nonumber \\
&\le \EX[\exp(\theta (\sum_{j,k:d_j(k)\le d+1} r_{j,k,s} \Bigg(\sum_{w=1}^{d} Q_{j,w}(\tau-1-d+w) + A_{j,1}(\tau-d_j(k):-d_j(k))\Bigg)))] \nonumber \\ 
&\qquad \times \EX[\exp(-\theta \sum_m \sum_{\tau\le t \le 0} U_m(t) h_{m,s'})]\mbox{.} \nonumber
\end{align*}

We use the previously derived bound for $\EX[\exp(-\theta \sum_m 
\sum_{\tau\le t \le 0} U_m(t) h_{m,s'})]$. So, we only concern
ourselves with 
\begin{align*}
&\EX[\exp(\theta (\sum_{j,k:d_j(k)\le d+1} r_{j,k,s} \bigg(\sum_{w=1}^{d} Q_{j,w}(\tau-1-d+w) + A_{j,1}(\tau-d_j(k):-d_j(k))\bigg) ))]\mbox{.} 
\end{align*}
Consider any $Q_{j,w}(\tau-1-d+w)$ at depth $w$, then $A_{j,1}(\tau-d-1)$ is
independent of it. As $A_{j,1}$ are i.i.d.\ and future arrivals in a queue
are independent of present and past queue-lengths, we have
\begin{align*}
& \EX[\exp(\theta (\sum_{j,k:d_j(k)\le d+1} r_{j,k,s} 
(\sum_{w=1}^{d} Q_{j,w}(\tau-1-d+w) + A_{j,1}(\tau-d_j(k):-d_j(k)))] \nonumber \\
&= \EX[\exp(\theta(\sum_{j,k:d_j(k)\le d+1} r_{j,k,s}\sum_{w=1}^{d} Q_{j,w}(\tau-1-d+w)))] \times \EX[\exp(\sum_{j,k:d_j(k)\le d+1} r_{j,k,s} A_{j,1}(\tau-d_j(k):-d_j(k)))] \nonumber
\end{align*}

For the second term we obtain a bound using previous techniques and note that
since $\vlambda \in \alpha \mc{C}$, 
\[
\sum_{j,k:d_j(k)\le d+1} r_{j,k,s} \EX[A_{j,1}] \le \sum_m \mu_m h_{m,s'}\mbox{,}
\]
which in the same way as above will imply that for some $K(N)$ and some $\theta>0$,
\begin{align*}
& \EX[\exp(\theta (\sum_{j,k:d_j(k)\le d+1} r_{j,k,s} A_{j,1}(\tau-d_j(k):-d_j(k)) - \sum_m \sum_{\tau\le t \le 0} U_m(t) h_{m,s'}))] \nonumber \\
&\quad\le \exp(-\theta K(N) \tau)\mbox{.}
\end{align*}

Note that 
\[
\EX[\exp(\theta(\sum_{j,k:d_j(k)\le d+1} r_{j,k,s}\sum_{w=1}^{d} Q_{j,w}(\tau-1-d+w)))] < \infty
\]
by the inductive assumption that the number of unallocated steps at depth $\le d$ have
finite exponential moments.

So we have that 
\[
\EX[\exp(\theta \sum_{\tau\le t \le 0}(X_s(t) - Y_s(t) + \bar{r}))]< c_1 \exp(-\theta K(N) \tau)\mbox{,}
\]
and so, in turn (using the same steps as above) $Q^s$ has finite exponential moment for 
some $\theta$. The rest of the steps are similar to above and we get the desired result
that 
\[
\EX[\exp(\theta \sum_{j,k: d_j(k) \le d+1} Q_{j,k})] < \infty\mbox{.}
\]
By induction on $d$, we have proven that the total number of unallocated steps over all types of jobs
have finite exponential moment (say $c'$).

Therefore,
\begin{align*} 
\PR(\sum_{j,k} Q_{j,k} > q) &\le \exp(-\theta q) \EX[\exp(\theta \sum_{j,k} Q_{j,k})] \nonumber\\
&\le c' \exp(-\theta q)\mbox{.} \nonumber
\end{align*}
So for $q = \frac{3\log N}{\theta}$, we have the result (as $c'$ is constant). 

\subsection{Proof of Theorem \ref{thm:LPrelaxFlex}}

The following lemma is useful for the proof.
\begin{lemma}
\label{lem:equivalentOpti}
Any feasible solution of problem \eqref{eq:RvarOpti} is a feasible solution of problem \eqref{eq:MvarOpti}.
\end{lemma}
\begin{IEEEproof}
Among the total $R$ available agents, let $R_m$ be of type $m$,
and let them be denoted $i_1, i_2, \dots, i_{R_m}$.
For an allocation $\mbf{a}$ in the formulation
\eqref{eq:RvarOpti}, let $\tilde{h}_{i_k,s}$ be the 
time that agent $i_1$ is assigned for skill $s$. Then among all these type $m$ agents,
the total contribution to skill $s$ is $\sum_{i=1}^{R_m} \tilde{h}_{i_k,s}$.

Note that the total time for these $R_m$ agents is $R_m h_m$. Now by the allocation
constraint we have $\sum_s \sum_{i=1}^{R_m} \tilde{h}_{i_k,s} \le R_m h_m$.
If we choose $\alpha_{m,s} \ge \sum_{i=1}^{R_m} \tilde{h}_{i_k,s}/R_m h_m$,
this is a valid allocation in formulation \eqref{eq:MvarOpti}, as 
$\sum_s \alpha_{m,s} \le 1$ and it also meets the allocation constraint. 
So, for this $\alpha_{m,s}$ the allocation $\mbf{a}$ is a valid
allocation in problem \eqref{eq:MvarOpti}.

This proves that every valid allocation in \eqref{eq:RvarOpti} is also
a valid allocation in \eqref{eq:MvarOpti}.
\end{IEEEproof}

Since solving \eqref{eq:MvarOpti} yields a feasible allocation, the lemma implies the two problems are actually 
alternate formulations of one another.
The rest of the proof follows the same steps as the proof of Thm.~\ref{thm:LPrelax}.

\subsection{Proof of Theorem \ref{thm:flexGreedy}}
The result can be derived in the same way as Thm.~\ref{thm:prioGreedy}, through the
use of the following lemmata.

\begin{lemma}
\label{lem:flexGreedy1}
Let $\{R_i\}$ be i.i.d.\ Bernoulli random variables with $\PR(R_i=1)=p \in (0,1)$ and
$N$ be a random variable independent of $\{R_i\}$ with a moment generating function
$M_N(\theta)$. Then
\[ 
\EX[\exp(\theta \sum_{i=1}^N R_i)] =  M_N(\log\left(p\exp(\theta) + (1-p)\right))\mbox{.}
\]
\end{lemma}
\begin{IEEEproof}
\begin{align*}
\EX[\exp(\theta \sum_{i=1}^N R_i)] &=\EX\left[\EX[\exp(\theta \sum_{i=1}^N R_i)|N]\right] \nonumber \\
&= \EX\left[\left(p\exp(\theta) + (1-p)\right)^N\right] \nonumber \\
&= \EX\left[\exp(\log\left(p\exp(\theta) + (1-p)\right) N) \right] \nonumber \\
&= M_N(\log\left(p\exp(\theta) + (1-p)\right))
\end{align*}
\end{IEEEproof}

\begin{lemma}
\label{lem:outerFF}
Let $\mc{C}_{F,F}^O=\{\vlambda: \exists b_m \in [0,1]^S \mbox{ for all } m,
\mbox{s.t.} \ \sum_{s} b_{m,s} \le 1, b_{m,s}>0$ \ \mbox{only if}
$s \in S_m, \sum_{m: s \in S_m} b_{m,s} h_m \mu_m > \sum_{j,k} \lambda_j r_{j,k,s}
\mbox{ for all } s\}$. Then,
$\mc{C}_{F,F} \subset \mc{C}_{F,F}^O$.
\end{lemma}
\begin{IEEEproof}
This follows from the constraints in \eqref{eq:MvarOpti}. Because constraints
in \eqref{eq:MvarOpti} are per sample realization, and the above constraints are 
in expectation. So for constraints in \eqref{eq:MvarOpti} to be satisfied, the
above constraints must be satisfied.
\end{IEEEproof}

\begin{lemma}
\label{lem:flexGreedy2}
For any $\tau$, $\PR\{\cap_{t=\tau}^\infty \cap_s \{\psi_{m,s} (t)=p_{m,s}
\}\}=1$ such that $p_{m,s}$ solves \eqref{eq:flexGreedyAlg1}, assuming ties
between multiple solutions are broken deterministically.\footnote{Extends to
random tie breaking also, but involves more details.}
\end{lemma}
\begin{IEEEproof}
For the choice of $\gamma(t)$ (and $t_0 = -\infty$) 
it follows from the convergence of stochastic 
approximation update equations \cite{Borkar2008} and the facts
that $\vlambda \in \alpha \mc{C}$ and $t_0=-\infty$. 

By stochastic approximation updates $\bar{A}$ and $\bar{U}$ converges
almost surely to $\vlambda$ and $\mbf{\mu}$ respectively at any finite $\tau$.
The rest follows from Lem.~\ref{lem:outerFF}.
\end{IEEEproof}

Let $\mc{B}_{s,m}$ be the set of agents of type $m$ that has been put into $\mc{B}_s$.
$\mc{B}_{s,m}$ is Bernoulli sampling from $U_m$ agents with probability 
$\psi_{m,s}(t)$.

We follow the same steps as in the proof of Thm.~\ref{thm:prioGreedy}. 
Consider the work and service time brought at time $t$ (as before). 
Note that work brought for skill $s$ is $\sum_j A_{j,1}(t) r_{j,1,s}$ and
service time brought by agents for skill $s$ is 
$\sum_{m: s \in S_m}  h_m \mc{B}_{s,m}(t)$. Hence, we can construct
a queue $\tilde{Q}_1$ (as before)
\begin{align}
Q^s_1(t+1) = (Q^s_1(t) + S \max_s \sum_j A_{j,1}(t) r_{j,1,s} - S \min_s \sum_{m: s \in S_m}  h_m \mc{B}_{s,m}(t))^+ \nonumber
\end{align}

Following the same steps to obtain \eqref{eq:prioGreedy2} we can have
\begin{align}
\EX[\exp(\theta \tilde{Q}_1)] \le \sum_{\tau \le 0} \EX[\exp(\theta S \sum_{\tau \le t \le 0}
(\max_s \sum_j A_{j,1}(t) r_{j,1,s} - \min_s \sum_{m: s \in S_m}  h_m \mc{B}_{s,m}(t))] \nonumber
\end{align}
But a result like \eqref{eq:prioGreedy2} does not follow immediately,
since $\mc{B}_{s,m}(t)$ are not independent over time and $\mc{B}_{s,m}(t)$
depends on $A_j(t)$ via $\{\psi_s(t)\}$. 

Consider the following. Let $X_s(t)=\sum_j A_{j,1}(t) r_{j,1,s}$ and
$Y_s(t)=\sum_{m: s \in S_m}  h_m \mc{B}_{s,m}(t)$, then
\begin{align}
&\EX[\exp(\theta S \sum_{\tau \le t \le 0}
(\max_s X_s(t) - \min_s Y_s(t))] \nonumber \\
&\quad\le \EX[\EX[\exp(\theta S \sum_{\tau \le t \le 0}
(\max_s X_s(t) - \min_s Y_s(t)) |\{\psi_{m,s}(t'),s,m t' \ge \tau\}]] \nonumber
\end{align}

Now by Lem.~\ref{lem:flexGreedy2}, for any finite $\tau_1$,
\[
\PR\{\{\psi_{s}(t')=p_s, \forall t' \ge \tau_1, \forall s,m\} = 1\mbox{.}
\]

Hence, for any finite $\tau_1, \tau_2$, 
$\{\psi_{m,s}(t'),s,m t' \ge \tau_1\}$ are independent of 
$\{A_{j,1}(t): t \tau_2\}$. 

Also, by the above argument, 
$X_s(t)$ and $Y_s(t)$ are independent of each other,
given $\{\psi_{m,s}(t'),s,m t' \ge \tau\}$ and they are also
independent over time. 
\begin{align*}
&\EX[\prod_{\tau \le t \le 0}\EX[\exp(\theta S (\max_s X_s(t) - \min_s Y_s(t)) |\{\psi_{m,s}(t'),s,m t' \ge \tau\})]] \nonumber \\
&= \EX[\sum_{\tau \le t \le 0} \EX[\exp(\theta S \max_{s,s'} (X_s(t)-Y_{s'}(t)) |\{\psi_{m,s}(t'),s,m t' \ge \tau\})]] \nonumber \\
&= \EX[\sum_{\tau \le t \le 0} \sum_{s,s'} \EX[\exp(\theta S (X_s(t)-Y_{s'}(t)) |\{\psi_{m,s}(t'),s,m t'\ge \tau\})]] \nonumber \\
&= \sum_{\tau \le t \le 0} \sum_{s,s'} \EX[\EX[\exp(\theta S X_s(t)|\{\psi_{m,s}(t'),s,m t'\ge \tau\})] \nonumber \\
& \quad\times \EX[\exp(-\theta S Y_{s'}(t)|\{\psi_{m,s}(t'),s,m t'\ge \tau\})]] \nonumber \\
& = \sum_{\tau \le t \le 0} \sum_{s,s'} \EX[\exp(\theta S X_s(t)|\{\psi_{m,s}(t')=\psi_{m,s}(\tau),s,m t'\ge \tau\})]  \nonumber \\
& \quad\times \EX[\exp(-\theta S Y_{s'}(t)|\{\psi_{m,s}(t')=\psi_{m,s}(\tau),s,m t'\ge \tau\})] \nonumber \\
& = \sum_{\tau \le t \le 0} \sum_{s,s'} \EX[\exp(\theta S X_s(t)] \times \EX[\exp(-\theta S Y_{s'}(t)|\{\psi_{m,s}(t')=\psi_{m,s}(\tau),s,m t'\ge \tau\})]
\end{align*}
The first equality follows because $\max_{x \in \mc{X}} f(x) - \min_{x \in \mc{X}}
g(x) = \max_{x,x' \in \mc{X}}(f(x)-g(x))$, for finite $\mc{X}$.
The second equality follows due to independence of $X_s(t)-Y_s(t)$ from
$t\ge \tau$ which is due to Lem.~\ref{lem:flexGreedy2}. The third
equality follows due to independence of $X_s(t)$ and $Y_s(t)$ given
$\{\psi_s(t)\}$. The fourth equality is again due to Lem.~\ref{lem:flexGreedy2},
as $\{\psi_{s}(t')=p_s,s,m t'\ge \tau\}$ is an almost sure event. The last
equality follows because $A_{j}(t), t\ge \tau_1$ are independent of 
$\{\psi_{s}(t')=p_s,s,m t'\ge \tau_2\}$ for any finite $\tau_1$ and $\tau_2$.

Note that $\EX[\exp(\theta S X_s(t)|]$ can be evaluated exactly as in
the proof of Thm.~\ref{thm:prioGreedy}.

Consider for $\tau \le t \le 0$,
\begin{align}
& \EX[\exp(-\theta S Y_{s'}(t)|\{\psi_{m,s}(t')=\psi_{m,s}(\tau),s,m t'\ge \tau\})] \nonumber \\
& = \EX[\exp(-\theta S \sum_{m: s \in S_m} h_m \mc{B}_{s,m}|\{\psi_{m,s}(t')=\psi_{m,s}(\tau),s,m t'\ge \tau\})] \nonumber \\
& = \prod_{m: s \in S_m} \EX[\exp(-\theta S h_m \mc{B}_{s,m}|\{\psi_{m,s}(t')=\psi_{m,s}(\tau),s,m t'\ge \tau\})] \label{eq:flexGreedy6}.
\end{align}

Now $\mc{B}_{s,m}$ is Bernoulli sampling of $U_m$ agents with probability
$\psi_{m,s}(t')=\psi_{m,s}(\tau):=p_{m,s}$. 

Then, by Lem.~\ref{lem:flexGreedy1}, for a $\tilde{\theta}_{m,s} = \log(\exp(\theta) p_{m,s} + 1-p_{m,s})$:
\begin{align*}
\EX[\exp(-\theta S h_m \mc{B}_{s,m}|\{\psi_{s}(t')=p_s,s,m t'\ge \tau\})] = \EX[\exp(-\tilde{\theta}_{m,s} h_m U_m)] \mbox{.}
\end{align*}

Now following the same steps as in the proof of Thm.~\ref{thm:prioGreedy}, we can
show that (for Poisson and Gaussian dominated cases) for a sufficiently large
$N_{\alpha}$, for all $N\ge N_{\alpha}$, and $\delta<\frac{\alpha}{2}$, 
\[
\EX[\exp(-\tilde{\theta}_{m,s} h_m U_m)] \le \exp(-(1-\delta)\tilde{\theta}_{m,s} h_m \mu_m)\mbox{.}
\]

Now, by concavity of logarithms,
\begin{align*}
\tilde{\theta}_{m,s} &= \log\left(p_{m,s}\exp(\theta) + (1-p_{m,s})\right) \nonumber \\
&\ge p_{m,s} \theta \mbox{.}
\end{align*}
This implies
\begin{equation}
\EX[\exp(-\tilde{\theta}_{m,s} h_m U_m)] \le \exp(-(1-\delta) \theta \ p_{m,s} h_m \mu_m)\mbox{.}
\label{eq:flexGreedy5}
\end{equation}

As $\vlambda \in \alpha \mc{C}_{F,F}$, by
Lem.~\ref{lem:outerFF} and \ref{lem:flexGreedy2}, we have
$\sum_{m: s \in S_m} p_{m,s} h_m \mu_m > (1-\alpha)
\sum_{j,k} \lambda_j r_{j,k,s}$ for all
$s$ and $\epsilon < 1-\alpha$. Also, as by assumption 
$|\sum_{j,k} \lambda_j r_{j,k,s} - \sum_{j,k} \lambda_j r_{j,k,s'}|$
is sub-poly$(N)$, hence
$$\sum_{m: s \in S_m} p_{s',m} h_m \mu_m > (1-\alpha)
\sum_{j,k} \lambda_j r_{j,k,s}, \forall s,s'$$

This along with \eqref{eq:flexGreedy5} and \eqref{eq:flexGreedy6} gives the final
result by following the same steps as the proof of Thm.~\ref{thm:prioGreedy}.

\subsection{Proof of Theorem \ref{thm:restGreedy}}
The assumption $\{S_m: m \in [M]\}$ is a partition implies that there exists
a partition of $[S]$, say $\{\mc{K}_l:1 \le l \le L\}$ such that
for $m \in [M]$, $S_m = \mc{K}_l$ for some $l \in [L]$.
Note that $L \le S$. 

As $\vlambda \in \mc{C}_{F,I}$, for any step $(j,k)$ with $\lambda_j>0$,
the set of required skills is a subset of some $\mc{K}_l$. Otherwise,
due to inflexibility of the steps, that step can never be allocated which
contradicts that $\vlambda \in \mc{C}_{F,I}$.

\begin{lemma}
\label{lem:outerFI}
Let $\mc{C}^O_{F,I} = \{\vlambda: \mbox{ for all } l \ \sum_{j:(j,k)-\mbox{skills} 
\ \subset \mc{K}_l} \lambda_j < \sum_{m: S_m = \mc{K}_l} \mu_m 
\lfloor \frac{h_m}{r_{1,1,1}}\rfloor \}$. Then, under the conditions
in Thm.~\ref{thm:restGreedy}, $\mc{C}_{F,I} 
\subset \mc{C}^O_{F,I}$.
\end{lemma}
\begin{IEEEproof}
Follows by noting the fact that under the conditions in Thm.~\ref{thm:restGreedy} the steps with skill requirements in
$\mc{K}_l$ has can only be served by agents of type $m$
with $S_m=\mc{K}_l$. Also, note that agents with skills
in $\mc{K}_l$ cannot serve any other kinds of steps.

Also, as each step is of same size $r_{1,1,1}$ a type $m$
agent can serve at most $\lfloor \frac{h_m}{r_{1,1,1}}\rfloor$
steps with skill requirements in $S_m$.

This implies that if the agent availability is $\bu$ then $a_{j,k}$ steps
can be served only if
\[
\sum_{j:(j,k)-\mbox{skills} \ \subset \mc{K}_l} a_{j,k} < \sum_{m: S_m = \mc{K}_l} u_m \lfloor \frac{h_m}{r_{1,1,1}}\rfloor \mbox{.}
\]
The rest follows because the sample path constraint is true only if the constraint is true in expectation.
\end{IEEEproof}

Notice that the condition \[
\sum_{j:(j,k)-\mbox{skills} \ \subset \mc{K}_l} \lambda_j < \sum_{m: S_m = \mc{K}_l} \mu_m  \lfloor \frac{h_m}{r_{1,1,1}}\rfloor \}
\] 
can be written as
\[
\sum_d \sum_{j:(j,k)\mbox{-skill} \ \subset \mc{K}_l\cap \mc{A}_d} \lambda_j < \sum_{m: S_m = \mc{K}_l} \mu_m 
\lfloor \frac{h_m}{r_{1,1,1}}\rfloor\mbox{.}
\]
This will be useful later.

Coming back to the main proof, we consider a queue for each $l \in [L]$ and depth $d$ ($\le D$), 
$\tilde{Q}^l_d(t)$. This queue represents the number of unallocated steps with
skill requirements in $\mc{K}_l$ that are at depth $\le d$ on
respective precedence trees.

Note that such steps with skill requirements in $\mc{K}_l$ 
can only be served by agents of types
$\{m: S_m = \mc{K}_l\}$. Note that just as in Restricted Greedy,
steps at depth $\le d$ have priority (allocate themselves before) 
over steps at higher depth:
\begin{align}
\tilde{Q}^l_d(t+1) = (\tilde{Q}^l_d(t) + \sum_{(j,k)\mbox{-skills} \subset \mc{K}_l} A_{j,k}(t)  - \sum_{m: S_m = \mc{K}_l} \lfloor \frac{h_m}{r_{1,1,1}}\rfloor
U_m)^+\mbox{.}
\end{align}

Note that allocation of steps with skill requirements in $\mc{K}_l$ and $\mc{K}_{l'}$
for $l\neq l'$ are independent. Also, the agents with skills in $\mc{K}_l$ and 
$\mc{K}_{l'}$ for $l\neq l'$ are independent. So, if we define $\tilde{Q}_d(t)$:
\begin{align}
\tilde{Q}_d(t+1) = & (\tilde{Q}_d(t) + L \max_l (\sum_{(j,k)\mbox{-skills} \subset
\mc{K}_l} A_{j,k}(t) - \sum_{m: S_m = \mc{K}_l} \lfloor \frac{h_m}{r_{1,1,1}}\rfloor
U_m))^+\mbox{,}
\end{align}
then $\tilde{Q}_d(t)$ is a path-wise upper-bound on $\sum_l \tilde{Q}^l_d(t)$.

Consider the depth $d=1$ first. Then we can follow the same steps
as in the proof of Thm.~\ref{thm:prioGreedy} for the depth $d=1$ case.
Note that $\vlambda \in \alpha \mc{C}_{F,I}$ implies that for all $l$
\[
\sum_{d=1}^D \sum_{j:(j,k)\mbox{-skill}
\ \subset \mc{K}_l\cap \mc{A}_d} \lambda_j < (1-\alpha) \sum_{m: S_m = \mc{K}_l} \mu_m 
\lfloor \frac{h_m}{r_{1,1,1}}\rfloor\mbox{,}
\]
and hence, for any $d\le D$
\[
\sum_{d'=1}^d \sum_{j:(j,k)\mbox{-skill}
\ \subset \mc{K}_l\cap \mc{A}_{d'}} \lambda_j < (1-\alpha) \sum_{m: S_m = \mc{K}_l} \mu_m 
\lfloor \frac{h_m}{r_{1,1,1}}\rfloor\mbox{.}
\]

This along with the same steps as in the proof of Thm.~\ref{thm:prioGreedy}
gives that for some $\theta_1>0$ and $\theta_1=\Omega(1)$, $\forall \theta<\theta_1$
\[
\EX[\exp(\theta_1 \tilde{Q}_1)] < \infty, \mbox{ for all } l\mbox{.}
\]

Then, like the proof of Thm.~\ref{thm:prioGreedy} we can perform induction over $d$ to
prove that for some $\theta_D>0$ and $\theta_D = \Omega(1)$, 
then for all $\theta<\theta_D$:
\[
\EX[\exp(\theta \tilde{Q}_D)] < \infty, \mbox{ for all } l\mbox{,}
\]
where $\tilde{Q}_D(t)$ is an upper bound on $\sum_l Q^l_d(t)$, which is again
the total number of unallocated steps in the system. We obtain
the desired bound from this.

Induction from $d$ to $d+1$ is similar to the proof of Thm.~\ref{thm:prioGreedy}.

\end{appendices}
\end{document}